\documentclass{article}%
\usepackage{amssymb}
\usepackage{amsfonts}
\usepackage{amsmath}
\usepackage{graphicx}
\usepackage{appendix}%
\providecommand{\U}[1]{\protect\rule{.1in}{.1in}}
\providecommand{\U}[1]{\protect\rule{.1in}{.1in}}
\providecommand{\U}[1]{\protect\rule{.1in}{.1in}}
\providecommand{\U}[1]{\protect\rule{.1in}{.1in}}
\newtheorem{theorem}{Theorem}

\newtheorem{corollary}[theorem]{Corollary}

\newtheorem{definition}[theorem]{Definition}
\newtheorem{example}[theorem]{Example}
\newtheorem{exercise}[theorem]{Exercise}
\newtheorem{lemma}[theorem]{Lemma}

\newtheorem{problem}[theorem]{Problem}
\newtheorem{proposition}[theorem]{Proposition}
\newtheorem{remark}[theorem]{Remark}
\newtheorem{solution}[theorem]{Solution}

\newenvironment{proof}[1][Proof]{\noindent\textbf{#1.} }{\ \rule{0.5em}{0.5em}}
\begin{document}

\title{Revisiting the Schwarzschild and the Hilbert-Droste Solutions of Einstein
Equation and the Maximal Extension of the Latter\thanks{Preliminary version.}}
\author{Igor Mol\thanks{igormol@ime.unicamp.br. Undergraduate Mathematics student at
IMECC-UNICAMP.}\\Institute of Mathematics, Statistics and Scientific Computation\\Unicamp, SP, Brazil}
\maketitle

\begin{abstract}
In this pedagogical note, the differences between the Schwarzschild and the
Hilbert-Droste solutions of Einstein equation are scrutinized through a
rigorous mathematical approach, based on the idea of warped product of
manifolds. It will be shown that those solutions are indeed \emph{different}
because the topologies of the manifolds corresponding to them are different.
After establishing this fact beyond any doubt, the maximal extension of the
Hilbert-Droste solution (the Kruskal-Szekeres spacetime) is derived with
details and its topology compared with the ones of the Schwazschild and the
Hilbert-Droste solution.

We also study the problem of the imbedding of the Hilbert-Droste solution in a
vector manifold, hopefully clarifying the work of Kasner and Fronsdal on the subject.

In an Appendix, we present a rigorous discussion of the Einstein-Rosen Bridge.
A comprehensive bibliography of the historical papers involved in our work is
given at the end.

\newpage

\end{abstract}
\tableofcontents

\section{Introduction}

The journal \textit{General Relativity and Gravitation }reprinted in 2003 the
famous paper in which Schwarzschild consecrated himself as the first person to
find an exact solution of the Einstein field equation (cf. ref. \cite{OG}).
Following the same volume of that journal, S. Antoci and D.-E. Liebscher
published an editorial note claiming that the solution presented by
Schwarzschild in 1916 (which describes the gravitational field generated by a
point of mass) is not equivalent to the one currently taught in textbooks on
General Relativity. The latter being a solution which was, however, found by
J. Droste and D. Hilbert just a year after Schwarzschild's publication. This
event culminated in a series of papers concerned with the equivalence or the
nature of these two solutions.

Three years after the editorial note of Antoci and Liebscher, a rectification
was published in the above journal (cf. ref. \cite{Sen}) claiming that the
solutions of Schwarzschild and of Hilbert-Droste are indeed equivalent, based
on the existence of a coordinate transformation for which the metric found
originally by Schwarzschild can be written in the same coordinate form as the
one found by Droste and Hilbert. This opinion is shared by the authors of ref.
\cite{Bel}, published in 2007, and of ref. \cite{Fro}, published in 2013.

However, the latter authors ignored that a spacetime is not only defined by a
metric, but also by the topology of the corresponding manifold. And in fact,
as we shall explain in details later, while the Schwarzschild manifold is
homeomorphic to $\mathbb{R}\times\left]  0,\infty\right[  \times S^{2}$,
leaving no room for a black hole and dispensing a procedure of maximal
extension, the topology of the Hilbert-Droste manifold is homeomorphic to
$\mathbb{R}\times\left(  \left]  0,\infty\right[  -\{\mu\}\right)  \times
S^{2}$ (for some real $\mu>0$), being consequently a \textit{different}
solution of the Einstein equation. We remark that the latter solution having a
disconnected manifold require a maximal extension in order to become a
satisfactory spacetime (cf. Definition \ref{Spacetime}).

This was recognized by N. Stavroulakis in his writings entitled
\textit{\textquotedblleft Math\'{e}matiques et trous noirs\textquotedblright}
(cf. ref. \cite{Nik}), which appeared in the \emph{Gazette des
math\'{e}maticiens}, and \textquotedblleft\textit{V\'{e}rit\'{e} scientifique
et trous noirs\textquotedblright\ (cf. refs.} \cite{Nik3}--\cite{Nik6}),
published just four years before the Antoci \& Liebscher editorial note. (We
shall comment briefly on Stavroulakis's articles in the final section).
Another author, who seems to be one of the first to advocate that the
solutions of Schwarzschild and of Hilbert-Droste are really different, was L.
Abrams, publishing about the subject already in 1979 (cf. ref. \cite{Abr}).

It is important to remark that because the Hilbert-Droste solution has a
disconnected topology (which as we will show below, is not the case of the
manifold in Schwarzschild's solution), the Relativity community was lead to
the \textquotedblleft Maximal Extension\textquotedblright\ research programme,
which grown from a J. Synge's letter to the editor in a Nature's volume which
dates from 1949, and culminated in the Kruskal-Szekeres spacetime and in the
Fronsdal imbedding of the Hilbert-Droste manifold -- a procedure which was
based in a work of E. Kasner from 1921 (almost four decades before Fronsdal's
paper was published). This, of course, inaugurated the physics of black holes.

Our paper revisit this issues from a mathematically rigorous standpoint and is
organized as follows. In Section \ref{MF}, we present the mathematical
formalism which will be adopted in rest of our work. In particular, we discuss
the warped product of manifolds, which is a powerful tool in constructing
spacetimes in General Relativity, some issues concerning the extension of
manifolds (which is complemented by the Appendix \ref{Exten}) and the
properties of null (or lightlike) geodesics which are useful in verifying that
a given manifold is maximal. \ In Section \ref{Solut}, we set a framework in
which both the Schwarzschild and the Hilbert-Droste solutions can be
constructed, in such a way that a parallel between their derivations and the
origin of their topological differences will be shown.

In Section \ref{Exten-HD}, motivated by the disconnectedness of the
Hilbert-Droste solution, we begin the search for its maximal extension,
covering details normally omitted by the present literature leading to the
Kruskal-Szekeres spacetime. A \textit{brief} summary of the relevant
historical developments is then presented. Lastly, we proceed to discuss the
works of Kasner and Fronsdal that culminated in the embedding of the
Hilbert-Droste spacetime in a 6-dimensional vectorial manifold, thus ending
this chapter in the history of General Relativity.

Finally, in Section \ref{Final}, we restate our main conclusions and comment
on some works in the literature. And, in Appendix \ref{ERB}, we give a short
but rigorous discussion of the Einstein-Rosen Bridge and some of its
mathematical relations to the Horizon that belongs to the Kruskal-Szekeres spacetime.

\section{Mathematical Formalism\label{MF}}

In order to fix our notation and refresh the memory, we review in Subsections
\ref{Mani} and \ref{ST} some elementary facts concerning pseudo-Riemannian
geometry, Minkowski vector spaces and spacetimes.

Then, the following two subsections are dedicated to a discussion of the
\textit{warped product}, a powerful tool that can be employed in the
construction of some spacetimes in General Relativity. As we shall see, its
use has at least two advantages: it can elegantly simplify calculations
related to geometric quantities, as the Ricci curvature tensor, and even more
important, when a spacetime is given in the form of a warped product, its
manifold topology is stated without ambiguities since the beginning.

Finally, in the Subsection \ref{Light}, we discuss some properties of null
geodesics which shall be useful (cf. Section \ref{Exten-HD}) in our
construction of the maximal extension of the Hilbert-Droste solution (the
Kruskal-Szekeres spacetime), a subject which is normally treated very
informally in the current literature.

In Appendix \ref{Exten}, our discussion of the extension of manifolds is
continued from a topological point of view. There, we discuss some topological
issues which may arises when two topological spaces are glued together through
a continuous identification of its topological subspaces. That Appendix is
however unnecessary for our main developments, but will be used in the
rigorous construction of the Einstein-Rosen bridge presented in Appendix
\ref{ERB}.

\subsection{Manifolds and Exponential Mapping\label{Mani}}

First, recall that

\begin{definition}
\label{pRm}A pseudo-Riemannian manifold is an ordered pair $(M,g)$, where
$M$\textit{\ }is a smooth manifold and $g\in\sec T_{2}^{0}M$ is a metric
tensor, i.e., a symmetric and non degenerate $2$-covariant tensor field in $M$
with the same index in all tangent spaces of $M$. We may say that $M$ have a
pseudo-Riemannian structure.
\end{definition}

Remember that the \textit{index} of a symmetric bilinear form $g$ is the
greatest integer $\upsilon$ such that there is a subspace $W$ with the
properties: $\dim W=\upsilon$ and $g(x,x)<0$ for all $x\in W$.

When there is no fear of confusion, we may refer to a pseudo-Riemannian
manifold $(M,g)$ just by $M$.

\begin{definition}
Let $M$ be a pseudo-Riemannian manifold and let $\gamma$ be a curve from
$I\subset\mathbb{R}$ into $M$. Let $\hat{D}_{\gamma}$ be the induced
Levi-Civita connection \emph{(}of $g$\emph{)} on $\gamma$. So we will call
$\gamma$ a geodesic if $\hat{D}_{\gamma}\gamma^{\prime}(t)=0$ for all $t\in I$.
\end{definition}

In what follows, unless we use the adjective \textit{segmented}, all geodesics
are defined on a interval which contains $0\in\mathbb{R}$.

Recall that a geodesic $\gamma$ defined on $I\subset\mathbb{R}$ is called
inextendible if and only if, for all geodesics $\sigma$ defined on
$J\subset\mathbb{R}$ such that $\sigma^{\prime}(0)=\gamma^{\prime}(0)$, we
have that $J\subset I$. To each $x\in T_{p}M$, we will denote by $\gamma_{x}$
the unique inextendible geodesic such that $\gamma_{x}^{^{\prime}}(0)=x$.

The idea of approximate the neighborhood of a point in a manifold through the
tangent space in that point can be made precise by using the exponential mapping:

\begin{definition}
\label{Exponential}Let $M$ be a pseudo-Riemannian manifold and let $p\in M$.
Let $D_{p}$ be the subset of $T_{p}M$ such that, for all $x\in D_{p}$, the
domain of $\gamma_{x}$ contains $[0,1]\subset\mathbb{R}$. The exponential
mapping $\exp_{p}$ at $p$ is the mapping from $D_{p}$ into $M$ such that
$x\rightarrow\exp_{p}(x)=\gamma_{x}(1)$.
\end{definition}

\begin{remark}
\label{expSmooth}Let $\gamma$ be a geodesic with induced Levi-Civita
connection $\hat{D}_{\gamma}$. As, in coordinates, $\hat{D}_{\gamma}%
\gamma^{\prime}(t)=0$ corresponds to a system of ordinary differential
equations of second order, the solution depends smoothly on the initial
values. Then the exponential mapping is a well-defined smooth mapping.
\end{remark}

In this paragraph, to each $\theta\in T_{p}^{\ast}M$, we will denote by
$\mathrm{d}\mathbf{\theta}$ the differential mapping of $\theta$ as being a
function from $T_{p}M$ into $\mathbb{R}$, and \emph{not} the exterior
derivative of $\theta$ as being a \emph{covector field}. In the proof of the
following Lemma, given $x\in M$, the natural homomorphism $\phi$ between
$T_{x}(T_{p}M)$ and $T_{p}M$ is the mapping such that, for all covector
$\theta\in T_{p}^{\ast}M$, $\theta\left[  \phi(v_{x})\right]  =\mathrm{d}%
\mathbf{\theta(}v_{x})$, for all $x\in T_{p}M$.

\begin{lemma}
\label{NormalDiffeo}Let $M$ be a pseudo-Riemannian manifold. For each $p\in M
$, there is a neighborhood $V\subset T_{p}M$ of $0\in T_{p}M$ such that
$\exp_{p}|V$ is a diffeomorphism.
\end{lemma}

\begin{proof}
Let $\phi$ be the natural homomorphism between $T_{x}(T_{p}M)$ and $T_{p}M$.
Let $v_{0}\in T_{0}(T_{p}M)$, let $v=\phi(v_{0})$ and let $\lambda(t)=vt$ be a
mapping from $\mathbb{R}$ into $T_{p}M$. So, as $\lambda^{\prime}(0)=v_{0}$,%
\[
\exp_{p\ast}(v_{0})=\exp_{p\ast}\left[  \lambda^{\prime}(0)\right]  =\left(
\exp_{p\ast}\circ\lambda\right)  ^{\prime}(0)=v
\]
Hence $\exp_{p\ast}$ is the natural homomorphism $\phi$. By Remark
\ref{expSmooth} and the inverse mapping theorem, the result follows.
\end{proof}

\begin{definition}
Let $M$ be a pseudo-Riemannian manifold and let $p\in M$. A neighborhood $U$
of $p$ will be called normal if there is a neighborhood $V\subset T_{p}M$ of
$0\in T_{p}M$ such that $\exp_{p}|V$ is a diffeomorphism between $V$ and $U$
and, for all $x\in V$, $\{tx:t\in\lbrack0,1]\subset\mathbb{R}\}\subset V$.
\end{definition}

So the last Lemma ensures that we can always find a normal neighborhood for a
given point.

\begin{lemma}
\label{Normalneigh}Let $M$ be a pseudo-Riemannian manifold, let $p\in M$ and
let $U$ be a normal neighborhood of $p$. So, for all $q\in U$, there is a
unique geodesic $\gamma_{pq}$ from $[0,1]\subset\mathbb{R}$ into $U$ such that
$\gamma_{pq}(0)=p$, $\gamma_{pq}(1)=q$ and $\gamma_{pq}^{\prime}(0)=\exp
_{p}^{-1}(q)$.
\end{lemma}

\begin{proof}
Let $v=\exp_{p}^{-1}(q)$ and let $\lambda(t)=vt$ be a mapping from
$\mathbb{R}$ into $T_{p}M$. Let $\sigma(t)=\exp_{p}\circ$ $\lambda(t)$ be a
mapping from $[0,1]\subset\mathbb{R}$ into $U$. By the hypothesis on $V$,
$\sigma$ is well-defined, and by the Definition \ref{Exponential}, $\sigma$ is
a geodesic. But%
\[
\sigma^{\prime}(0)=\left(  \exp_{p\ast}\circ\lambda\right)  ^{\prime}%
(0)=\exp_{p\ast}\left[  \lambda^{\prime}(0)\right]  =v
\]
by the proof of the last Lemma. Hence the existence assertion. The proof of
the uniqueness will be left as an easy exercise.
\end{proof}

\bigskip

Let $\gamma$ be a curve from $[a,b]\subset\mathbb{R}$ into a pseudo-Riemannian
manifold $M$. We will say that $\gamma$ is a broken-geodesic if there is a
partition $(J_{i})_{i\in F\subset\mathbb{N}}$ of $[a,b]$ such that each
restriction $\gamma|J_{i}$, for $i\in F$, is a segmented geodesic. In this
case, we say that $\gamma(a)$ and $\gamma(b)$ are connected by a broken-geodesic.

\begin{corollary}
\label{Geodesics}A pseudo-Riemannian manifold $M$ is connected if and only if,
for all points $p,q\in M$, there exists a broken-geodesic $\gamma$ defined on
$[a,b]\subset\mathbb{R}$ such that $\gamma(a)=p$ and $\gamma(b)=q$.
\end{corollary}

\begin{proof}
Let $S$ be the subset of $M$ of all points that can be connected by a
broken-geodesic and let $p\in M$. Let $U$ be a normal neighborhood of $p$. So,
by Lemma \emph{\ref{Normalneigh}}, if $p\in S$, $U\subset S$. But if $p\notin
S$, then $U\cap S=\varnothing$, and $M$ cannot be connected. Hence the result.
\end{proof}

\bigskip

In what follows, we will call a neighborhood $U$ in a pseudo-Riemannian
manifold \textit{convex} if $U$ is a normal neighborhood for all $p\in U$. To
see a proof that a convex neighborhood always exists around any given point,
see Chapter 5 of \cite{On1}.

\subsection{Spacetimes\label{ST}}

Spacetimes are the manifolds upon which the General Relativity Theory is
established. To define them, we need to recall some facts about Lorentz vector spaces:

\begin{definition}
A Lorentz vector space\ is an ordered par $(V,g)$, where $V$ is a
finite-dimensional linear space with dimension $\dim V\geq2$ and $g$ is a
symmetric and non degenerate bilinear form on $V$ with index $1$.
\end{definition}

A sequence $(e_{i})_{i\in F\subset\mathbb{N}}$ of vectors in a given Lorentz
vector space $(V,g)$ will be called orthonormal if $\left\vert g(e_{i}%
,e_{j})\right\vert =\delta_{ij}$, where $\delta_{ij}$ is the Kronecker delta
(i.e., $\delta_{ij}=0$ when $i\neq j$ and $\delta_{ii}=1$).

\begin{lemma}
\label{Orthonormal}Let $(V,g)$ be a Lorentz vector space. So there is an
orthonormal basis for $V$.
\end{lemma}

\begin{proof}
\emph{(i)} As $g$ is non degenerate, there is a $x\in V$ such that
$g(x,x)\neq0$. \emph{(ii)} If $(e_{i})_{i\in\lbrack1,k]}$ is a sequence of
orthonormal vectors (for some $k<\dim V$), there is a vector $e_{k+1}$ such
that $(e_{i})_{i\in\lbrack1,k+1]}$ is also orthonormal, by \emph{(i)} and by
the fact that $g$ is non degenerated in the subspace $\{x\in V:g(x,e_{i}%
)=0,i\in\lbrack1,k]\subset\mathbb{N}\}$. The result follows then by induction.
\end{proof}

\begin{definition}
Let $(V,g)$ be a Lorentz vector space. A vector $x\in V$ will be called
timelike if $g(x,x)<0$, spacelike if $g(x,x)>0$ and null (or lightlike) if
$g(x,x)=0$. A vector is causal if it is timelike or null. A subspace $W\subset
V$ is called timelike, spacelike or null if all vectors in $W$ are timelike,
spacelike and null, respectively.
\end{definition}

On what follows, given a Lorentz vector space $(V,g)$, the orthogonal
complement of $x\in V$ is the subset $x^{\perp}=\{z\in V:g(x,z)=0\}$. The
reader may prove that $x^{\perp}$ is, in fact, a subspace.

Let $(e_{i})_{i\in\lbrack1,n]}$ be an orthonormal basis for a n-dimensional
Lorentz vector space $(V,g)$ and let $(\varepsilon_{i})_{i\in\lbrack1,n]}$ be
a sequence numbers such that $g(e_{i},e_{j})=\varepsilon_{i}\delta_{ij}$. For
the proof of the next Lemma, recall \cite{Bour} that the Sylvester Theorem
ensures that there is one and only one $k\in\lbrack1,n]\subset\mathbb{N}$ such
that $\varepsilon_{k}=-1$.

\begin{lemma}
\label{Lorentz-decomp}Let $(V,g)$ be a Lorentz vector space and let $x\in V$.
So $x^{\perp}$ is timelike \emph{(}respectively, spacelike\emph{)} if $x$ is
spacelike \emph{(}respectively, timelike).
\end{lemma}

\begin{proof}
Let $n=\dim V$ and suppose that $x$ is timelike. By the proof Lemma
\emph{\ref{Orthonormal}}, there is an orthonormal sequence $(e_{i}%
)_{i\in\lbrack1,n-1]}$ of vectors in $V$ such that $(e_{i})_{i\in\lbrack1,n]}$
is an orthonormal basis for $V$, where $e_{n}=x/\sqrt{g(x,x)}$. Let $y\in
x^{\perp}$. So there is a sequence $(a_{i})_{i\in\lbrack1,n]}$ of real numbers
such that $y=\sum_{i\in\lbrack1,n]}a_{i}e_{i}$. By hypothesis, $a_{n}=0$.
Hence, by Sylvester Theorem, $g(y,y)=\sum_{i\in\lbrack1,n]}(a_{i})^{2}>0$,
i.e., $y$ is spacelike, and the proof is analogous if $x$ is spacelike.
\end{proof}

\bigskip

From now on, the set of all timelike vectors in a given Lorentz vector space
$(V,g)$ will be denoted by $\tau$, while that the set of all null vectors will
be denoted by $\Lambda$. These are normally called, respectively, the
\textit{timecone} and the \textit{lightcone} of $V$. The union $\tau
\cup\Lambda$ will be called the \textit{causalcone} and denoted by $\Upsilon$.

\begin{exercise}
Using Lemma\emph{ \ref{Lorentz-decomp}}, prove that the timecone, lightcone
and the causalcone of a given Lorentz vector space have two disjoint
components. Also, prove that the closure of a component of the timecone is a
component of the lightcone. \emph{(}For details, see Chapter 5 of
\emph{\cite{On1}} or Chapter 1 of \emph{\cite{Sac})}.
\end{exercise}

Then we shall denote by $\tau^{+}$ and $\tau^{-}$ the disjoint components of
the timecone $\tau$, and by $\Lambda^{+}$ and $\Lambda^{-}$ their respective
boundaries (which are, of course, the disjoint components of $\Lambda$). The
closure of $\tau^{+}$ and $\tau^{-}$, which will be denoted by $\Upsilon^{+}$
and $\Upsilon^{-}$, respectively, are the components of the causalcone
$\Upsilon$.

\begin{lemma}
\label{sign}Let $(V,g)$ be a Lorentz vector space and let $x\in\tau^{+}$. So
$y$ $\in\Upsilon^{+}$ if and only if $g(x,y)<0$ and $z$ $\in\Upsilon^{-}$ if
and only if $g(x,z)>0$.
\end{lemma}

\begin{proof}
Let $f$ be the continuos mapping from $\Upsilon$ into $\mathbb{R}-\{0\}$ such
that $v\rightarrow f(v)=g(x,v)$. As $f(x)<0$ and $\Upsilon^{+}$ is connected,
$f(\Upsilon^{+})=(-\infty,0)\subset\mathbb{R}$. If $z$ $\in\Upsilon^{-}$, thus
-$z$ $\in\Upsilon^{+}$, hence the result.
\end{proof}

\bigskip

Now, we are ready to generalize this to a manifold:

\begin{definition}
A Lorentzian manifold is an orientable 4-dimensional pseudo-Riemannian
manifold whose index of the metric is \emph{1}.
\end{definition}

Given a Lorentzian manifold $M$, let $\pi$ be the natural projection from $TM$
onto $M$. An element $x\in TM$ will be called timelike, spacelike and null (or
lightlike) if $x$, as an element of the Lorentz vector space $T_{\pi(x)}M$, is
timelike, spacelike or null, respectively. As before, $x\in TM$ is causal if
it is timelike or null.

Let $\gamma$ be a curve from $I\in\mathbb{R}$ into a pseudo-Riemannian
manifold $M$ and let $\hat{D}_{\gamma}$ be the induced Levi-Civita connection
on $\gamma$. For the proof of the following Lemma, remember that, given a
vector field $X\in\sec T\gamma$ over $\gamma$, we say that $X$ is parallel if
$\hat{D}_{\gamma}X=0$. Let $x\in T_{\gamma(a)}M$ for some $a\in I$. By the
theory of differential equations, there is one and only one parallel vector
field $X\in\sec T\gamma$ such that $X_{a}=x$. In this case, $y\in
T_{\gamma(b)}M$ (for some $b\in I$) will be called the parallel transport
(from $a$ to $b$) of $x$ along $\gamma$ if $X_{b}=y$.

\begin{lemma}
\label{timelemma}Let $M$ be a connected Lorentzian manifold. The subset
$\tau(M)\subset TM$ of all causal vectors is connected or have two components.
\end{lemma}

\begin{proof}
Let $p\in M$ and let $A$ be the set of all broken-geodesics in $M$. By the
Corollary \ref{Geodesics} and the axiom of choice, there is a mapping
$\delta_{p}$ from $M$ into $A$ such that each $\delta_{p}(q)$ is a
broken-geodesic from $p$ into $q$. Let $\Upsilon_{r}^{+}$ and $\Upsilon
_{r}^{-}$ be the components of the causalcone $\Upsilon_{r}\subset T_{r}M$ for
any $r\in M$. To each $q\in M$, define$\ \hat{\Upsilon}_{q}^{+},\hat{\Upsilon
}_{q}^{-}\subset T_{q}M$ to be such that $\hat{x}\in$ $\hat{\Upsilon}_{q}^{+}$
and $\hat{y}\in\hat{\Upsilon}_{q}^{-}$ if and only if there exists $x\in$
$\Upsilon_{p}^{+}$ and $y\in\Upsilon_{p}^{-}$ such that $\hat{x}$ and $\hat
{y}$ are, respectively, the parallel transport of $x$ and $y$ along
$\delta_{p}(q)$. Hence, by virtue of the Levi-Civita connection, $\hat{x}%
\in\Upsilon_{q}^{+}$ and $\hat{y}\in\Upsilon_{q}^{-}$, and $\tau(M)=\cup_{q\in
M}\left(  \hat{\Upsilon}_{q}^{+}\cup\hat{\Upsilon}_{q}^{-}\right)  $.
Consequently, $\tau(M)$ have at most two components, and the result follows.
\end{proof}

\begin{problem}
\label{TimeProblem}Let $M$ be a connected Lorentzian manifold and let
$X\in\sec TM$ be a smooth timelike vector field, i.e., $X_{p}\in T_{p}M$ is a
timelike vector for all $p\in M$. So the subset $\tau(M)\subset TM$ of all
causal vectors have two components.
\end{problem}

\begin{solution}
Let $g$ be the metric of $M$ and let $f$ be the continuos mapping from
$\tau(M)$ onto $\mathbb{R}-\{0\}$ such that $V_{p}\rightarrow f(V_{p}%
)=g(X_{p},V_{p})$. As $f(X_{p})<0$ for all $p\in M$, $f^{-1}(-\infty,0)$ and
$f^{-1}(0,\infty)$ must be two disconnected components, and the result follows
from Lemma \emph{\ref{timelemma}}.
\end{solution}

In the case of the last Problem, we usually say that a vector $Y\in TM$ is
future-pointing if it is in the same component of $\tau(M)$ as $X$.

Finally,

\begin{definition}
A connected Lorentzian manifold is time-orientable if and only if the subset
$\tau(M)\subset TM$ of all causal vectors has two components.
\end{definition}

\begin{definition}
\label{Spacetime}A spacetime (in General Relativity) is a connected orientable
ad time-orientable Lorentzian manifold $(M,g)$ equipped.with the Levi-Civita
connection $D$ of $g$.
\end{definition}

\begin{remark}
\label{TimeMotivation}A physical motivation for the last Definition is that,
if we assume that the thermodynamics holds for any process in a given
spacetime, it must be possible to select a \textquotedblleft time
arrow\textquotedblright\ for the physical phenomena from the second law, given
a time orientation in that spacetime.
\end{remark}

\subsection{Product of Manifolds}

In what follows, given two manifolds $M$ and $N$, the natural projections
$\pi_{M}$ and $\pi_{N}$\ of $M\times N$ are the mappings from $M\times N$ into
$M$ and $N$, respectively, such that $\pi_{M}(p,q)=p$ and $\pi_{N}(p,q)=q$.

\begin{lemma}
Let $(M,g_{M})$ and $(N,g_{N})$ be pseudo-Rimannian manifolds and let $\pi
_{M}$ and $\pi_{N}$ be the natural projections of $M\times N$. Then $(M\times
N,g)$, where%
\[
g=\pi_{M}^{\ast}(g_{M})+\pi_{N}^{\ast}(g_{N})
\]
is itself a pseudo-Rimannian manifold, called the product manifold of
$(M,g_{M})$ and $(N,g_{N})$.
\end{lemma}

The proof is a direct application of Definition \ref{pRm} and will be left as
an easy exercise.

In order to transport mappings, vectors and tensors from manifolds $M$ and $N
$ to the product manifold $M\times N$, the notion of a lift will be introduced
below. For the sake of brevity, consider the following notation:
\[
T_{(p,q)}M=T_{(p,q)}(M\times\{q\})
\]%
\[
T_{(p,q)}N=T_{(p,q)}(\{p\}\times N)
\]
for all $(p,q)\in M\times N$.

\begin{lemma}
\label{Decomp}Let $M$ and $N$ be smooth manifolds. So to each $(p,q)\in
M\times N$, $T_{(p,q)}(M\times N)$ is the direct sum of $T_{(p,q)}M$ and
$T_{(p,q)}N$.
\end{lemma}

\begin{proof}
By definition, $\pi_{M}|(\{p\}\times N)$ is a constant function. So
$\pi_{M\ast}(T_{(p,q)}N)=\{0\}$. But $\pi_{M\ast}|T_{(p,q)}M$ is an
isomorphism onto $T_{p}M$. Hence $T_{(p,q)}M\cap T_{(p,q)}N=\{0\}$. The result
follows then by $\dim T_{(p,q)}(M\times N)=\dim T_{(p,q)}M+\dim T_{(p,q)}N$.
\end{proof}

Because of the identifications between $T_{(p,q)}M$ and $T_{p}M$ and between
$T_{(p,q)}N$ and $T_{q}N$, one normally recall the last Lemma in applications
as saying that $T_{(p,q)}(M\times N)=(T_{p}M)\times(T_{q}N)$.

Hereafter, given a manifold $M$, the set of all smooth mappings from $M$ into
$\mathbb{R}$ will be denoted by $F(M)$.

\begin{definition}
Let $M$ and $N$ be smooth manifolds let $\pi_{M}$ and $\pi_{N}$ be the natural
projections of $M\times N$. We define the lifts in $M\times N$ of the mappings
$f\in F(M)$ and $g\in F(N)$ to be the functions $\mathbf{f=}f\circ\pi_{M}$ and
$\mathbf{g=}g\circ\pi_{N}$, respectively. We also define the lifts in $M\times
N$ of the vectors $x\in T_{p}M$ and $y\in T_{p}N $ as the unique
$\mathbf{x\in}$ $T_{(p,q)}M$ and $\mathbf{y\in}$ $T_{(p,q)}N$, respectively,
such that $\pi_{M\ast}\mathbf{x}=x$ and $\pi_{N\ast}\mathbf{y}=y$.
\end{definition}

\begin{remark}
The uniqueness assertion in the last definition is ensured by Lemma
\emph{\ref{Decomp}}.
\end{remark}

We can extrapolate the above definition to vector fields in the following way:

\begin{definition}
\label{Lift1}Let $M$ and $N$ be smooth manifolds. Let $X\in\sec TM$ and
$Y\in\sec TN$ be a vector fields. We define the lifts in $M\times N$ of $X$
and $Y$ to be the unique vector fields $\mathbf{X,Y}$ $\in\sec T(M\times N)$
such that $\mathbf{X}_{p}$ is the lift in $M\times N$ of $X_{p}\in T_{p}M$ and
$\mathbf{Y}_{p}$ is the lift in $M\times N$ of $Y_{p}\in T_{p}N$. We will say
that $\mathbf{X}$ is a horizontal lift in $M\times N$, while that $\mathbf{Y}$
is a vertical lift.
\end{definition}

\begin{remark}
Using coordinates, one can prove that the lift of a smooth vector field is by
itself smooth.
\end{remark}

\begin{example}
In $\mathbb{R}^{2}$ with natural coordinates $(x,y)$, $\frac{\partial
}{\partial x}$ is the horizontal lift of $\frac{d}{dt}$, while that
$\frac{\partial}{\partial y}$is the vertical one.
\end{example}

From now on, in the terminology of Definition \ref{Lift1}, the set of all
horizontal lifts in $M\times N$ will be denoted by $\pounds (M)$, whereas the
set of all vertical lifts will be denoted by $\pounds (N)$.

Finally, we need to define the lift of a $r$-covariant tensor field:

\begin{definition}
Let $M$ and $N$ be smooth manifolds and let $\pi_{M}$ and $\pi_{N}$ be the
natural projections of $M\times N$. Let $A\in\sec T^{r}M$ and $B\in\sec
T^{r}N$ be r-covariant tensor fields. We define the lifts in $M\times N$ of
$A$ and $B$ to be the unique r-covariant tensor fields $\mathbf{A,B}\in\sec
T^{r}(M\times N)$ such that, for all $(p,q)\in M\times N$ and $(v_{i}%
)_{i\in\lbrack1,r]}\in T_{(p,q)}(M\times N)$, $\mathbf{A(}v_{1},...,v_{r}%
)=A(\pi_{M\ast}(v_{1}),...,\pi_{M\ast}(v_{r}))$ and $\mathbf{B(}%
v_{1},...,v_{r})=B(\pi_{N\ast}(v_{1}),...,\pi_{N\ast}(v_{r}))$.
\end{definition}

\begin{remark}
\emph{(a)} Using Lemma \emph{\ref{Decomp}}, one can prove the uniqueness
assertion. \emph{(b)} This definition cannot be used to lift an arbitrary
$(s,r)$-tensor field, since that $\pi_{M}^{\ast}$ and $\pi_{M\ast}$ goes in
"opposite" directions. But using Definition \emph{\ref{Lift1}}, the reader is
invited to inquire how to lift a $(s,1)$-tensor field.
\end{remark}

\subsection{Warped Product}

In General Relativity, many spacetimes can be constructed in the following way:

\begin{definition}
\label{Warped}Let $(B,g_{B})$ and $(F,g_{F})$ be pseudo-Rimannian manifolds
and $\pi_{M}$ and $\pi_{N}$ be the natural projections of $B\times F$. Let $f$
be a smooth mapping from $F$ into $\mathbb{R}^{+}$ (the set of positive real
numbers). We define the warped product $B\times_{f}F$ to be the
pseudo-Rimannian manifold $(B\times F,g)$ such that%
\[
g=\pi_{B}^{\ast}(g_{B})+(f\circ\pi_{F})^{2}\pi_{F}^{\ast}(g_{F})
\]
The function $f$ may be called the warping mapping of $B\times_{f}F$.
\end{definition}

\begin{example}
\label{HomeoR3}Let $r$ be the identity mapping in $\mathbb{R}^{+}$ and let
$(\phi,\varphi)$ be polar coordinates in $S^{2}=S^{2}(1)$. Let
\[
\eta=d\phi\otimes d\phi+\sin^{2}\phi d\varphi\otimes d\varphi
\]
be the Euclidean metric in $S^{2}$. Then $\mathbb{R}^{+}\times_{r}S^{2}$ is
\textit{isometric to} the Euclidean space $\mathbb{R}^{3}-\{0\}$.
\end{example}

\begin{exercise}
Let $n$ be a positive integer and let $v\in\lbrack0,n)\subset\mathbb{N}.$ Let
$(x^{i})_{i\in\lbrack0,n]}$ be the natural coordinates of $\mathbb{R}^{n+1}$.
So $(\mathbb{R}^{n+1},\zeta)$ is the pseudo-Euclidean n-space of index $v$
when%
\[
\zeta=-\sum_{i\in\lbrack1,v]}dx^{i}\otimes dx^{i}+\sum_{i\in\lbrack
v+1,n+1]}dx^{i}\otimes dx^{i}.
\]
Then the pseudo-Euclidean n-sphere $S_{v}^{n}$ of index $v$ is the $n$-sphere
$S^{n}\subset\mathbb{R}^{n+1}$ with the induced connection of $(\mathbb{R}%
^{n+1},\zeta)$. Show how $S_{v}^{n}$ can be written as a warped product of
$S^{n-v}$.
\end{exercise}

Recall that, given a smooth mapping $f$ from pseudo-Rimannian manifold $M$
(together with a metric tensor $g$) into $\mathbb{R}$, $\operatorname{grad}%
(f)$ is the vector field metric equivalent to $df$, that is,
\[
g(\operatorname{grad}(f),X)=df(X)=X(f)
\]
for all vector $X\in TM$. Then the Hessian of $f$ is defined to be the
2-covariant tensor field such that%
\[
(V,W)\rightarrow H^{f}(V,W)=VW(f)-(D_{V}W)=g(D_{V}(\operatorname{grad}(f)),W)
\]
and the Laplacian of $f$ is\ simply the contraction of $H^{f}$, i.e.,
$\Delta(f)=CH^{f}$.

The following Lemma will be our bridge between the geometry of $B$ and $F$ and
its warped product $B\times_{f}F$:

\begin{lemma}
\label{Ricci}With the notation of Definition \emph{\ref{Warped}}, let
$(M,g_{M})=B\times_{f}F$, let $Ric^{M}$ be the Ricci curvature tensor of $M$
and let $\mathbf{Ric}^{B}\in\pounds (B)$ and\textbf{\ }$\mathbf{Ric}^{F}%
\in\pounds (F)$ be the lifts of the Ricci tensors of $B$ and $F$,
respectively. Suppose that $\dim F>1$ and define the mapping
\[
\Im(f)=\frac{\Delta(f)}{f}+(\dim F-1)\frac{g_{M}(\operatorname{grad}%
(f),\operatorname{grad}(f))}{f^{2}}%
\]
from $M$ into $\mathbb{R}$. Hence, for all $\mathbf{X,Y\in}$ $\pounds (B)$ and
$\mathbf{V,W}\in\pounds (F)$,
\[
Ric^{M}(\mathbf{X,W})=0,
\]%
\[
Ric^{M}(\mathbf{V,W})=\mathbf{Ric}^{F}(\mathbf{V,W)-}\Im(f)g_{M}%
(\mathbf{V,W),}%
\]%
\[
Ric^{M}(\mathbf{X,Y})=\mathbf{Ric}^{B}(\mathbf{X,Y)-}\frac{\dim F}{f}%
H^{f}(\mathbf{X,Y}).
\]

\end{lemma}

The proof of this Proposition follows a tedious application of definitions and
will then be omitted. The interested reader may consult the Chapter 7 of
\cite{On1}.

\subsection{Null Geodesics and Maximal Extensions\label{Light}}

\begin{definition}
\label{MaximalManifold}A pseudo-Riemannian manifold $M$ will be called maximal
when, for all pseudo-Riemannian manifolds $N$ with the same dimension of $M$
for which $M$ is isometric to an open submanifold, $M=N$.
\end{definition}

Differently from the Riemannian case, we cannot use the Hopf-Rinow Theorem to
decide when our spacetime is maximal. However, we can do it by studying the
behavior of the null geodesics.

\begin{lemma}
\label{Null}Let $M$ be a spacetime and let $U$ be a convex neighborhood in $M
$. So for all points $p,q\in U$, there is one $r\in U$ such that the unique
geodesics from $p$ into $r$ and from $r$ into $q$ are nulls.
\end{lemma}

To prove this, we will use the Gauss Lemma and introduce some terminology first.

For the last of this section, let $(M,g)$ be a pseudo-Riemannian manifold, let
$p\in M$, let $x\in T_{p}M$ and let $\phi_{x}$ be the natural homomorphism
between $T_{x}(T_{p}M)$ and $T_{p}M$ (recall the comment above Lemma
\ref{NormalDiffeo}). In what follows, a vector $v\in T_{x}(T_{p}M)$ will be
called radial if there is a real $k\neq0$ such that $\phi(v)=kx$, and we will
denote just by $g$ the metric for both $T_{p}M$ and $T_{x}(T_{p}M)$.

\begin{lemma}
[Gauss Lemma]Let $(M,g)$ be a pseudo-Riemannian manifold and let $p\in M$. Let
$v,w\in T_{x}(T_{p}M)$ and suppose that $v$ is radial. Then%
\[
g(v,w)=g(\exp_{p\ast}v,\exp_{p\ast}w).
\]

\end{lemma}

\begin{proof}
Let $\lambda(t,r)=t\left[  \phi_{x}(v)+s\phi_{x}(w)\right]  $ be a mapping
from $\mathbb{R}\times\mathbb{R}$ into $T_{p}M$ and let $x(t,r)=\exp_{p}%
\circ\lambda(t,r)$ be a mapping from $\mathbb{R}\times\mathbb{R}$ into $U$. As
$(D_{1}\lambda)(1,0)=v$ and $(D_{2}\lambda)(1,0)=w$, we have%
\[
(D_{1}x)(1,0)=\exp_{p\ast}v\text{ \ \ }(D_{2}x)(1,0)=\exp_{p\ast}w.
\]
But, by the definition of the exponential mapping, $t\mapsto x(t,r)$ is a
geodesic. Hence $D_{1}^{2}x=0$ and $g(D_{1}x,D_{1}x)=g(\phi_{x}(v)+s\phi
_{x}(w),\phi_{x}(v)+s\phi_{x}(w))$. Thus%
\[
D_{1}g(D_{1}x,D_{2}x)=g(D_{1}x,D_{2}D_{1}x)=\frac{1}{2}D_{2}g(D_{1}%
x,D_{1}x)=g(\phi_{x}(w),\phi_{x}(v)+s\phi_{x}(w)),
\]
which implies%
\[
\left[  D_{1}g(D_{1}x,D_{2}x)\right]  (t,0)=g(\phi_{x}(v),\phi_{x}(w)).
\]
The result follows then from the fact that $g(D_{1}x(0,0),D_{2}x(0,0))=0$ and
an elementary calculation.
\end{proof}

\bigskip

From now on, the \textit{position vector field} $\mathrm{P}\in\sec T(T_{p}M)$
in $T_{p}M$ is defined to be the vector field such that $\mathrm{P}_{x}%
=\phi_{x}^{-1}(x)$, and the \textit{quadratic form} $Q_{p}$ in $T(T_{p}M)$ is
the mapping into $\mathbb{R}$ given by $Q_{p}(x)=g(x,x)$. Then we may write
that $Q_{p}=g(\mathrm{P},\mathrm{P})$.

\begin{exercise}
Let $\tilde{D}$ be the Levi-Civita connection on the vector space $T_{p}M$.
Prove that, if $\mathrm{P}$ is the position vector field, then $\tilde{D}%
_{v}\mathrm{P}=v$ for all $v\in T(T_{p}M)$.\emph{ (}Hint: if you feel lost,
appeal to coordinates\emph{)}.
\end{exercise}

\begin{lemma}
Let $M$ be a pseudo-Riemannian manifold and let $p\in M$. Let $\mathrm{P}$ and
$Q$ be the position vector field and the quadratic form in $T_{p}M$,
respectively. So%
\[
\operatorname{grad}Q=2\mathrm{P}%
\]

\end{lemma}

\begin{proof}
Let $v\in T_{x}(T_{p}M)$ for some $x\in T_{p}M$. Then:%
\[
g(\operatorname{grad}Q,v)=dQ(v)=v\left[  g(\mathrm{P},\mathrm{P})\right]
=2g(P,v)
\]
by the last exercise, and the proof is over.
\end{proof}

\bigskip

The exponential mapping can extend the position vector field and the quadratic
form over a normal neighborhood in the following way. We define the
(transported) \textit{position vector} \textit{field }$\mathbf{P}$ $\in\sec
TU$ to be the vector field over $U$ given by $\mathbf{P}$ $=\exp_{p\ast}P$,
and the (transported) \textit{quadratic form} $\mathbf{Q}$ to be the mapping
on $U$ such that $x\rightarrow\mathbf{Q(}x)=Q\circ\exp_{p}^{-1}(x)$.

\begin{lemma}
Let $U$ be a normal neighborhood in a given pseudo-Riemannian manifold $M$ and
let $\mathbf{P}$ and $\mathbf{Q}$ be the transported position vector field and
the transported quadratic form, respectively. So%
\[
\operatorname{grad}\mathbf{Q}=2\mathbf{P.}%
\]

\end{lemma}

\begin{proof}
Let $y\in T_{q}U$ for some $q\in U$. Then%
\[
g(\operatorname{grad}\mathbf{Q},y)=d\left(  Q\circ\exp_{p}^{-1}\right)
(y)=g(\operatorname{grad}Q,\exp_{p}^{-1}y)=2g(P,\exp_{p}^{-1}y\text{)}%
\]
and the result follows by the Gauss Lemma.
\end{proof}

\bigskip

\begin{proof}
[Proof of Lemma \ref{Null}]Let $\Upsilon_{q}^{+}$ and $\Upsilon_{q}^{-}$ be
the disjoint components of the causalcone of $T_{q}M$. As $U$ is a convex
neighborhood, there is a unique geodesic $\sigma$ from $p$ into $q$. Without
loss of generality, suppose that $\sigma^{\prime}(0)\in\Upsilon_{q}^{-}$ (or
in intuitive terms, $p$ is in the past of $q$). Let $\gamma$ be a null
geodesic defined on $I\subset\mathbb{R}$ such that $\gamma(0)=p$. Let
$\mathbf{P}$ and $\mathbf{Q}$ be the transported position vector field and the
transported quadratic form in $T_{q}M$, respectively. So%
\[
\left(  \mathbf{Q\circ}\gamma\right)  ^{\prime}(t)=d\mathbf{Q}\left[
\gamma^{\prime}(t)\right]  =2g(\mathbf{P}_{\gamma(t)},\gamma^{\prime}(t)).
\]
But $\mathbf{Q\circ}\gamma(0)=\mathbf{Q(}p)\geq0$, by hypothesis. If
$\mathbf{Q\circ}\gamma(0)=0$, the result follows trivially. Then suppose that
$\mathbf{Q\circ}\gamma(0)>0$. By Lemma \ref{sign} and by the Gauss Lemma, the
equation above shows that $\left(  \mathbf{Q\circ}\gamma\right)  ^{\prime
}(t)<0$. Hence, there is a $k\in I$ such that $\mathbf{Q\circ}\gamma(k)=0$.
Then let $r=\gamma(k)$ and the proof is over.
\end{proof}

\begin{lemma}
\label{Maximal}Let $M$ be a spacetime and let $N$ be an open submanifold of
$M$ with the induced connection. Assume that, if $\gamma$ is a null geodesic
from $I\subset\mathbb{R}$ into $M$ such that $\gamma(I)\cap N\neq\varnothing$,
then $\gamma(I)\subset N$. Hence $M=N$.
\end{lemma}

\begin{proof}
Suppose that $M\neq N$ and let $U$ be a convex neighborhood in $M$ such that
$U\cap\partial N\neq\varnothing$. By hypothesis, there are $p\in U-N$ and
$q\in U\cap N$. By the Lemma \ref{Null}, there is some $r\in U$ such that the
unique geodesics $\gamma_{pr}$ from $p$ into $r$ and $\gamma_{rq}$ from $r$
into $q$ are nulls. By hypothesis, $\gamma_{pr}$ lies on $N$. Hence $r\in N$.
So $\gamma_{rq}$ lies on $N$. Thus $q\in N$. Contradiction.
\end{proof}

\begin{remark}
Physically, the last Corollary means that a spacetime is maximal if one cannot
\textquotedblleft see\textquotedblright\ beyond it.
\end{remark}

\section{Schwarzschild and Hilbert-Droste Solutions\label{Solut}}

It is well-known that K. Schwarzschild \cite{Sch} was the first to find the
exact gravitational field of a point of mass in General Relativity. However, a
year later, the same problem was differently approached by D. Hilbert
\cite{Hil} and J. Droste \cite{Dro}, differences which will be discussed below.

In Section \ref{BM}, we shall build a \textquotedblleft
spacetime\textquotedblright\ model in which the Schwarzschild and
Hilbert-Droste are particular cases. Hence, we show in Section \ref{GS} how to
generate solutions from such a model and we illustrate with a simple example.

Finally, we present in the last two sections the derivation of the
Hilbert-Droste and Schwarzschild solutions and we finish by discussing if
these are actually the same or not.

\subsection{Building the Model \label{BM}}

Let $(t,r)$ be be the natural coordinates of $\mathbb{R}^{2}$ and let
$P\subset\mathbb{R}\times\mathbb{R}^{+}$ be an open submanifold. In what
follows, $(t,h)$ will be called \textit{special coordinates} of $P$ if and
only if there is a diffeomorphism $\phi$ from $r(P)$ into $\mathbb{R}$ such
that $h=\phi\circ r$.

\begin{definition}
\label{Model0}A Schwarzschild model is an ordered list $(P,(t,h),f,g,\alpha)$,
where $P\subset\mathbb{R}\times\mathbb{R}^{+}$ is an open submanifold, $(t,h)$
is some special coordinates of $P$ and $f,g,\alpha$ are smooth mappings from
$h(P)$ into $\mathbb{R}^{+}$ such that%
\[
\lim_{h\rightarrow\infty}f(r)=\lim_{h\rightarrow\infty}g(r)=1
\]

\end{definition}

\begin{definition}
\label{Plane}Let $M=(P,(t,h),f,g,\alpha)$ be a Schwarzschild model. We define
the corresponding Schwarzschild plane $\Pi_{M}$ to be the pseudo-Riemannian
manifold $(P,\zeta)$ such that%
\[
\zeta=-(f\circ h)dt\otimes dt+(g\circ h)dh\otimes dh
\]

\end{definition}

In building a manifold through the warped product, the first step is to study
the geometry of its parts. In our case, we start by

\begin{lemma}
\label{Coderivatives}Given a Schwarzschild model $M=(P,(t,h),f,g,\alpha)$, let
$D$ be the Levi-Civita connection of its Schwarzschild plane $\Pi_{M}$. Thus%
\[
D_{\partial_{t}}\partial_{\partial_{t}}=\frac{f^{\prime}(h)}{2g(h)}%
\frac{\partial}{\partial h}\text{ \ \ }D_{\partial h}\partial_{h}%
=\frac{g^{\prime}(h)}{2g(h)}\frac{\partial}{\partial h}%
\]%
\[
D_{\partial_{t}}\partial_{h}=D_{\partial h}\partial_{t}=-\frac{f^{\prime}%
(h)}{2f(h)}\frac{\partial}{\partial t}%
\]

\end{lemma}

\begin{proof}
As the dimension of $P$ is 2, a direct computation is viable. So let
$(\Gamma_{i,j}^{k})_{(k,i,j)\in\lbrack1,2]^{3}}$ be the Christoffel symbols.
We use the well known equation%
\[
\Gamma_{~ij}^{k}=\frac{1}{2}\sum_{m\in\lbrack1,2]}\eta^{km}\left(
\frac{\partial\zeta_{im}}{\partial x^{j}}+\frac{\partial\zeta_{jm}}{\partial
x^{i}}-\frac{\partial\zeta_{ij}}{\partial x^{m}}\right)
\]
where $\zeta_{ij}=\zeta(\partial/\partial x^{i},\partial/\partial x^{j})$ and
$x^{1}=t,x^{2}=h$. Thus, for example,%
\[
\Gamma_{\mathrm{~}11}^{1}=0
\]%
\[
\Gamma_{~11}^{2}=-\frac{1}{2}\zeta^{22}\frac{\partial\zeta_{11}}{\partial
h}=\frac{f^{\prime}(h)}{2g(h)}%
\]
and the identity for $D_{\partial_{_{t}}}\partial_{t}$ follows. The last two
will be left as an easy exercise.
\end{proof}

So now we define our spacetime model:

\begin{definition}
\label{Model}Let $S^{2}$ be the Euclidean $2$-sphere and let
$M=(P,(t,h),f,g,\alpha)$ be a Schwarzschild model. So the \emph{(}%
Schwarzschild-like\emph{)} spacetime $\mathbf{S}_{M}$ associated with $M$ is
the warped product%
\[
\Pi_{M}\times_{\alpha}S^{2}%
\]

\end{definition}

\begin{remark}
\label{Coord?}As a Schwarzschild-like spacetime $\mathbf{S}$ is a
pseudo-Riemannian manifold, it can have distinct representations as a
Schwarzschild model. Indeed, to each possible choice of special coordinates
$(t,h)$ of $P$, there are mappings $f,g,\alpha$ such that
$M=(P,(t,h),f,g,\alpha)$ implies $\mathbf{S}=\mathbf{S}_{M}$. The submanifold
$P$ of $\mathbb{R}\times\mathbb{R}^{+}$ \emph{(}see Definition \ref{Plane}%
\emph{)} is, in the other hand, fixed: it is a part of the manifold of
$\mathbf{S}$. We only introduced the notion of a \textquotedblleft
Schwarzschild model\textquotedblright\ because, in finding a solution to
Einstein equation \emph{(}see next section\emph{)}, it is important to keep a
track of the coordinate system which we are using.
\end{remark}

\begin{remark}
\emph{(a)} The spacetime in Definition\emph{ \ref{Model}} can be time oriented
by lifting the coordinate vector $\partial/\partial t$; for more in time
orientability, see Chapter \emph{5} of \emph{\cite{On1}}. \emph{(b)} The
traditional physical motivations for the last Definition are that its
corresponding spacetime is \textquotedblleft static\textquotedblright\ with
respect to the \textquotedblleft time\textquotedblright\ $t$ \emph{(}see
Chapter \emph{12} of \emph{\cite{On1}} for a rigorous definition\emph{)},
spherically symmetric and, as $h\rightarrow\infty$, $\Pi_{M}$ approach the
Minkowski \textquotedblleft plane\textquotedblleft\ \emph{(}see Chapter
\emph{1} of \emph{\cite{Sac})}.
\end{remark}

\subsection{Generating Solutions \label{GS}}

Recall that a spacetime obeys the Einstein field equation in vacuum if and
only if it is Ricci flat.

In the following Proposition we will use Lemma \ref{Ricci} to find the
restrictions that the Einstein equation imposes upon our spacetime model:

\begin{proposition}
\label{Einstein}Let $M=(P,(t,h),f,g,\alpha)$ be a Schwarzschild model. Its
spacetime $\mathbf{S}_{M}$ satisfies the Einstein field equation in vacuum if
and only if%
\[
K=\frac{\alpha^{\prime}(r)f^{\prime}(r)}{\alpha(r)f(r)g(r)}=\frac{2}%
{\alpha(r)g(r)}\left[  \alpha^{\prime\prime}(r)-\frac{g^{\prime}(r)}%
{2g(r)}\alpha^{\prime}(r)\right]
\]%
\[
\Im(\alpha)=\frac{1}{[\alpha(h)]^{2}}%
\]
where $K$ is the sectional curvature of the Schwarzschild plane of $M$, given
by%
\[
K=-\frac{1}{2\sqrt{f(r)g(r)}}\left[  \frac{f^{\prime}(r)}{\sqrt{f(r)g(r)}%
}\right]  ^{\prime}%
\]
and%
\[
\Im(\alpha)=\frac{1}{\alpha(h)}\left\{  \frac{\alpha^{\prime\prime}(h)}%
{g(h)}+\frac{\alpha^{\prime}(h)}{2g(h)}\left[  \frac{f^{\prime}(h)}%
{f(h)}-\frac{g^{\prime}(h)}{g(h)}\right]  +\frac{[\alpha^{\prime}(h)]^{2}%
}{g(h)\alpha(h)}\right\}
\]

\end{proposition}

In the proof of this Proposition, the following two Lemma will be used:

\begin{lemma}
\label{Help1}Let $(M,g)$ be a pseudo-Riemannian surface \emph{(}that is, a
pseudo-Riemannian manifold such that $\dim M=2$\emph{)}. Let $Ric$ be its
Ricci curvature tensor and let $K$ be its sectional curvature. Then%
\[
Ric=Kg
\]

\end{lemma}

\begin{proof}
Let $(u,v)$ be orthogonal coordinates in some neighborhood of $M$ (which
always exists since we can employ a frame field; see, e.g., Chapter 3 of
\cite{On1}) and let $R\ $be the Riemannian curvature tensor of $(M,g)$.
Remember that, if $x,y\in T_{p}M$ are linearly independent vectors (for some
$p\in M$),%
\[
K(x,y)=\frac{g(R_{xy}x,y)}{Q(x,y)}%
\]
where
\[
Q(x,y)=g(x,x)g(y,y)-g(x,y)^{2}%
\]
Since $M$ has dimension $2$, $K$ is a smooth mapping in $F(M)$. But by
definition%
\begin{align*}
Ric(x,x)  &  =\frac{g(R_{x.\partial_{u}}x,\partial_{u})}{g(\partial
_{u},\partial_{u})}+\frac{g(R_{x.\partial_{v}}x,\partial_{v})}{g(\partial
_{v},\partial_{v})}\\
&  =K_{p}\left[  \frac{Q(x,\partial_{u})}{g(\partial_{u},\partial_{u})}%
+\frac{Q(x,\partial_{v})}{g(\partial_{v},\partial_{v})}\right]
\end{align*}
Then the result follows by a direct substitution in the above identity, and
the details are left as an easy exercise.\medskip
\end{proof}

As usual, in the following Lemma the partial derivative $\partial f/\partial
x$ of a mapping $f$ will be denoted just by $f_{x}$.

\begin{lemma}
\label{Help2}Let $(M,g)$ be a pseudo-Riemannian surface with sectional
curvature $K$. Let $(u,v)$ be an orthogonal coordinate system over $M$, let
$e,g\in\sec F(M)$ be positive real-valued mappings and let $\varepsilon
_{1}^{2}=\varepsilon_{2}^{2}=1$ be real numbers such that $\varepsilon
_{1}e^{2}=g(\partial_{u},\partial_{u})$ and $\varepsilon_{2}g^{2}%
=g(\partial_{v},\partial_{v})$, where $\partial_{u}$ and $\partial_{v}$ are
the coordinate vectors of $(u,v)$. Therefore%
\[
K=-\frac{1}{eg}\left[  \varepsilon_{1}\left(  \frac{e_{v}}{g}\right)
_{v}+\varepsilon_{2}\left(  \frac{g_{u}}{e}\right)  _{u}\right]
\]

\end{lemma}

\begin{proof}
[Proof of Proposition \ref{Einstein}]Let $Ric^{\Pi_{M}}$ and $Ric^{S^{2}}$ be
the Ricci curvature tensors of the Schwarzschild plane $\Pi_{M}$ and of the
Euclidean 2-sphere $S^{2}$, respectively. By Lemma \ref{Ricci}, the Einstein
field equation in vacuum ($Ric=0$) is equivalent to%
\[
\mathbf{Ric}^{\Pi_{M}}(\mathbf{X,Y)=}\text{ }\frac{2}{\alpha(r)}H^{\alpha
}(X,Y)
\]%
\[
\mathbf{Ric}^{S^{2}}(\mathbf{V,W)=\Im(}\alpha\mathbf{)}g(\mathbf{V,W)}%
\]
for all $V,W\in\sec TS^{2}$ and $X,Y\in\sec TP$, where $\mathbf{V,W\in
\pounds (}S^{2})$ and $\mathbf{X,Y\in\pounds (}P)$ are their respective lifts.

By Lemma \ref{Coderivatives} it is%
\begin{align*}
H^{\alpha}(\partial_{t},\partial_{t})  &  =-\frac{f^{\prime}(h)}{2g(h)}%
\alpha^{\prime}(r),\\
H^{\alpha}(\partial_{h},\partial_{h})  &  =\alpha^{\prime\prime}%
(h)-\frac{g^{\prime}(h)}{2g(h)}\alpha^{\prime}(h).
\end{align*}
Let $K$ be the sectional curvature of $\Pi_{M}$. Using Lemma \ref{Help1}, we
find that%
\[
\frac{f^{\prime}(h)\alpha^{\prime}(h)}{f(h)g(h)\alpha(h)}=K=\frac
{2}{g(h)\alpha(h)}\alpha^{\prime\prime}(h)-\frac{g^{\prime}(h)}{[g(h)]^{2}%
\alpha(h)}\alpha^{\prime}(h)
\]
and the expression for $K$ is a direct use of Lemma \ref{Help2}.

Finally, Lemma \ref{Help1} gives that%
\[
\mathbf{Ric}^{S^{2}}(\mathbf{V,W)=}\frac{g(\mathbf{V,W)}}{[\alpha(r)]^{2}}%
\]
Hence the second equation of the Proposition follows. The last is only a
direct computation, and will be left as an exercise (for the Definition of
$\Im(\alpha)$, see Lemma \ref{Ricci}).
\end{proof}

\begin{problem}
\label{Uniquely1}Fix the submanifold $P\subset\mathbb{R}\times\mathbb{R}^{+}$.
\emph{(a)} Is a Schwarzschild model $M=(P,(t,h),f,g,\alpha)$ uniquely
determined by Proposition \emph{\ref{Einstein}}? \emph{(b)} Is the associated
spacetime $\mathbf{S}_{M}$ of $M$, which satisfies Einstein equation in
vacuum, uniquely determined?
\end{problem}

\begin{solution}
\emph{(a)} No. \emph{(b)} Yes. One can add to the Einstein field equation some
\textquotedblleft coordinate condition\textquotedblright\ in order to
determine $f,g,\alpha$ uniquely. After this, we have a pseudo-Riemannian
manifold \emph{(}in particular, a spacetime\emph{)} whose metric and
\emph{(}if $P$ was given\emph{)} topology is well-defined. If we use some
different \textquotedblleft coordinate condition\textquotedblright, we must
find another set of $\mathbf{f},\mathbf{g},\mathbf{\alpha}$, but this is
because we are using distinct coordinate systems \emph{(}see Remark
\emph{\ref{Coord?}).}
\end{solution}

Indeed, the reader must already known the \textquotedblleft
Schwarzschild\textquotedblright\ solution which is normally presented in the
current literature (see the next section). We illustrate in the following
example another possible choice for $f,g,\alpha$ which also satisfies
Proposition \ref{Einstein} and the conditions of Definition \ref{Model0}:

\begin{example}
\label{Uniquely2}Let $\alpha(h)=r+\mu$. Assume that the spacetime
$\mathbf{S}_{M}$ of a Schwarzschild model $M=(P,(t,h),f,g,\alpha)$ satisfies
Einstein equation in vacuum. So by the first equation in Proposition
\emph{\ref{Einstein}},
\[
\left[  f(h)g(h)\right]  ^{\prime}=0
\]
But by Definition \ref{Model0} \emph{(}recall the limit conditions),
$f(r)g(r)=1$\emph{)}. Then by the second pair of equations of the same
Proposition, one finds that%
\[
g^{\prime}(h)=\frac{g(h)\left[  1-g(h)\right]  }{r+\mu}%
\]
Solving this equation, we find as a possible solution%
\[
g(h)=\frac{h+\mu}{h-\mu}%
\]
Of course that $\lim_{r\rightarrow\infty}g(r)=1$, as we needed. Hence, in
terms of coordinates, the metric $\mathbf{S}_{M}$ reads:%
\[
-\frac{h-\mu}{h+\mu}dt\otimes dt+\frac{h+\mu}{h-\mu}dr\otimes dr+(r+\mu
)^{2}\zeta_{S^{2}}%
\]
where $\zeta_{S^{2}}$ is the Euclidean metric of $S^{2}$.
\end{example}

\subsection{Hilbert-Droste Solution\label{HDSolution}}

The lesson which we must take from Problem \ref{Uniquely1} and Example
\ref{Uniquely2} is that, in order to find some solution of Einstein equation,
one needs to impose some \textquotedblleft coordinate
condition\textquotedblright.

The way followed by Hilbert was very simple and elegant, and can be summarized
in the following definition:

\begin{definition}
\label{HDModel}A Schwarzschild model $M=(P,(t,h),f,g,\alpha)$ will be called a
Hilbert model if and only if the special coordinates $(t,h)$ of $P$ were
chosen such that $\alpha\circ h=\operatorname{id}_{\mathbb{R}}$.
\end{definition}

In what follows, let $(w_{1},w_{2},w_{3},w_{4})$ be the natural coordinates of
$\mathbb{R}^{4}$. So in Hilbert words \cite{Hil} (translation from \cite{Ant}),

\begin{quotation}
According to Schwarzschild, if one poses%
\begin{align*}
w_{1}  &  =r\cos\vartheta\\
w_{2}  &  =r\sin\vartheta\cos\varphi\\
w_{3}  &  =r\sin\vartheta\sin\varphi\\
w_{4}  &  =l
\end{align*}
the most general interval corresponding to these hypotheses is represented in
spatial polar coordinates by the expression%
\[
(42)~~F(r)dr^{2}+G(r)(d\vartheta^{2}+\sin^{2}\vartheta d\varphi^{2}%
)+H(r)dl^{2}%
\]
where $F(r)$, $G(r)$ and $H(r)$ are still arbitrary functions of $r$. If we
pose%
\[
r^{\ast}=\sqrt{G(r)}%
\]
we are equally authorized to interpret $r^{\ast}$, $\vartheta$ and $\varphi$
as spatial polar coordinates. If we substitute in $(42)$ $r^{\ast}$ for $r$
and drop the symbol $\ast$, it results the expression
\[
M(r)dr^{2}+r^{2}(d\vartheta^{2}+\sin^{2}\vartheta d\varphi^{2})+W(r)dl^{2}%
\]
where $M(r)$ and $W(r)$ means the two essentially arbitrary functions of $r$.
\end{quotation}

With the last definition, we are able to derive the Hilbert-Droste
\textit{metric}:

\begin{proposition}
\label{Hilbert0}Let $M=(P,(t,h),f,g,\alpha)$ be a Hilbert model. Its spacetime
$\mathbf{S}_{M}$ obeys Einstein equation in vacuum if and only if%
\[
f(h)=\frac{1}{g(h)}=1-\frac{\mu}{h}%
\]
for some real $\mu$.
\end{proposition}

\begin{proof}
By the first equation of Proposition \ref{Einstein}, we have (like in Example
\ref{Uniquely2}),%
\[
\left[  f(h)g(h)\right]  ^{\prime}=0.
\]
But by Definition \ref{Plane} (recall the limit conditions), $f(h)g(h)=1$.
Then by the second pair of equations of the same Proposition, one finds that%
\[
g^{\prime}(h)=\frac{g(h)\left[  1-g(h)\right]  }{h}%
\]
Hence, there is a real number $\mu$ such that%
\[
g(h)=\frac{1}{1-\mu/h}%
\]
and the proposition is proved.
\end{proof}

\bigskip

We do not have, however, the complete \textit{Hilbert-Droste solution}. We
only have its metric, which is just half the story. To have in hands a
\textit{proper solution}, we must set up a topology, which in this case means
to pick up some $P\subset\mathbb{R}\times\mathbb{R}^{+}$ (recall Definitions
\ref{Plane} and \ref{Model}).

Note that the largest submanifold of $\mathbb{R}\times\mathbb{R}^{+}$ in which
the mappings in the last Proposition are smooth is $\mathbb{R}\times
(\mathbb{R}^{+}-\{\mu\})$. Then, we are motivated to state the Hilbert-Droste solution:

\begin{definition}
[Hilbert-Droste solution]\label{Hilbert-Droste}Given a real number $\mu$, the
Hilbert-Droste solution $H(\mu)$ is the spacetime $\mathbf{S}$ for which there
is a Hilbert model $M=(P,(t,h),f,g,\alpha)$ such that\textbf{\ }%
$\mathbf{S=S}_{M}$,
\[
P=\mathbb{R}\times(\mathbb{R}^{+}-\{\mu\}),
\]
and%
\[
f(h)=\frac{1}{g(h)}=1-\frac{\mu}{h}.
\]

\end{definition}

So by Proposition \ref{Hilbert0}, the Hilbert-Droste solution obeys the
Einstein equation in vacuum.\medskip

What distinguish the coordinate expression for the metric in the above
Proposition and in Example \ref{Uniquely2} is the choice of coordinates.
However, are Example \ref{Uniquely2} and Definition \ref{Hilbert-Droste}
describing the same solution?

\begin{problem}
Let $\mu\in\mathbb{R}$ and let $M=(P,(t,h),f,g,\alpha)$ be as in Example
\ref{Uniquely2}. Choose $P\subset\mathbb{R}\times\mathbb{R}^{+}$ to be the
largest submanifold for which $f,g$ are smooth (and the corresponding metric
non degenerated). Is the spacetime $\mathbf{S}_{M}$ a Hilbert-Droste solution?
\end{problem}

\begin{solution}
Yes. But taking into account that such a metric have a singularity in $r=\mu$
we see that the largest possible $P$ is $\mathbb{R}\times(\mathbb{R}^{+}%
-\{\mu\})$, hence it has the same topology as Hilbert-Droste.
\end{solution}

\begin{remark}
Playing with Proposition \emph{\ref{Einstein}}, one can generate an infinite
set of metrics for a Schwarzschild-like spacetime which satisfies Einstein
equation. In principle, one can find a coordinate transformation which
transform these metric expressions into each other. However, if we are
presented with two spacetimes whose metric expressions can be transformed into
each other in some coordinate chart, it does not means that they are the same
solution: it is necessary to take care about the topology, which in the
approach of this paper depends on a submanifold $P\subset\mathbb{R}%
\times\mathbb{R}^{+}$.
\end{remark}

\subsection{Schwarzschild Solution}

In this paragraph, let $(u^{1},u^{2},u^{3},u^{4})$ be a coordinate system on a
given spacetime with metric $g$. When Schwarzschild found his solution in
1916, he used the following form of the Einstein field equations in vacuum%
\begin{gather*}
\sum_{k\in\lbrack1,4]}\frac{\partial\Gamma^{k}~_{ij}}{\partial u^{k}}%
+\sum_{(k,l)\in\lbrack1,4]^{2}}\Gamma_{~il}^{k}\Gamma_{~kj}^{l}=0,\\
\sqrt{-\det g}=1,
\end{gather*}
for all $(i,j)\in\lbrack1,4]^{2}$, where $(\Gamma_{~ij}^{k})_{(k,i,j)\in
\lbrack1,4]^{3}}$ are the Christoffel symbols and $\det g$ is the determinant
of the matrix whose elements are $g_{ij}=g(\partial/\partial u^{i}%
,\partial/\partial u^{j})$ (see \cite{Ein}). The second equation is such that
only unimodular coordinate transformations preserves the \textquotedblleft
mathematical form\textquotedblright\ of the field equations.

Schwarzschild started his work by setting the spacetime manifold to be
$\mathbb{R}\times\{\mathbb{R}^{3}-\{0\}\}$. As he wanted a spherically
symmetric solution, it was natural \ for him to introduce spatial polar
coordinates. But the transformation from the natural coordinates of
$\mathbb{R}^{3}$ to polar coordinates is not, of course, unimodular. In his
own words \cite{Sch} (translation from \cite{Sch1}):

\begin{quotation}
When one goes over to polar co-ordinates according to $x=r\sin\vartheta
\cos\phi$, $y=r\sin\vartheta\sin\phi$, $z=r\cos\vartheta$ \textit{(...)} the
volume element \textit{(...)} is equal to $r^{2}\sin\vartheta drd\vartheta
d\phi$, \textit{[so] }the functional determinant $r^{2}\sin\vartheta$ of the
old with respect to the new coordinates is different from 1; then the field
equations would not remain in unaltered form if one would calculate with these
polar co-ordinates, and one would have to perform a cumbersome transformation.
\end{quotation}

Then Schwarzschild proceeded in the following way (also from \cite{Sch}):

\begin{quotation}
However there is an easy trick to circumvent this difficulty. One puts:%
\[
x_{1}=\frac{r^{3}}{3}\text{, }x_{2}=-\cos\vartheta\text{, }x_{3}=\phi
\]
Then we have for the volume element: $r^{2}\sin\vartheta drd\vartheta d\phi=$
$dx_{1}dx_{2}dx_{3}$. The new variables are then polar co-ordinates with the
determinant 1. They have the evident advantages of polar co-ordinates for the
treatment of the problem, and at the same time \textit{(...)} the field
equations and the determinant equation remain in unaltered form.
\end{quotation}

The reader must take in mind that, in the Schwarzschild approach, the
coordinate condition need to solve the equations of Proposition \ref{Einstein}
(recall also Problem \ref{Uniquely1}) must satisfies the Einstein's
determinant equation.

However, thanks to the warped product, we do not need to concern with any
\textquotedblleft polar coordinates with determinant 1\textquotedblright%
\ here. Indeed, using the following definition, we can do the whole derivation
without any mention to the coordinates of $S^{2}$:

\begin{definition}
\label{Unimodular}A Schwarzschild model $M=(P,(t,h),f,g,\alpha)$ will be
called a unimodular model if and only if the special coordinates $(t,h)$ of
$P$ were chosen such that%
\[
f(h)g(h)\left[  \alpha(h)\right]  ^{4}=1
\]

\end{definition}

\begin{remark}
Let $\zeta_{S^{2}}$ be the Euclidean metric of $S^{2}$, let
$M=(P,(t,h),f,g,\alpha)$ be a unimodular model and let $g$ be the metric of
the spacetime $\mathbf{S}_{M}$ of $M$. In the notation of the first paragraph,
if we give a coordinate expression to $\zeta_{S^{2}}$ such that $\det
\zeta_{S^{2}}=1$,
\[
-\det g=f(h)g(h)\left[  \alpha(h)\right]  ^{4}\det\zeta_{S^{2}}=1
\]
we have the Schwarzschild \textquotedblleft original\textquotedblright%
\ coordinate condition.
\end{remark}

To see how Definition \ref{Unimodular} together with Proposition
\ref{Einstein} determine the warping mapping $\alpha$ up to \emph{two}
constants, we state the following Lemma:

\begin{lemma}
If the spacetime of a unimodular model $(P,(t,h),f,g,\alpha)$ satisfies the
Einstein field equation in vacuum, then there are real numbers $\lambda,\mu$
such that
\[
\alpha(h)=\lambda\left(  3h+\mu^{3}\right)  ^{1/3}%
\]

\end{lemma}

\begin{proof}
By hypothesis,%
\[
\left[  f(h)g(h)\right]  ^{\prime}\left[  \alpha(h)\right]  +4\left[
f(h)g(h)\right]  \alpha^{\prime}(h)=0
\]
So, using the first equation of Proposition \ref{Einstein} we get%
\[
\alpha^{\prime\prime}(h)=-2\frac{[\alpha^{\prime}(h)]^{2}}{\alpha(h)}%
\]
whose solution is%
\[
\alpha(h)=\lambda\left(  3h+\mu^{3}\right)  ^{1/3}%
\]
for $\lambda,\mu\in\mathbb{R}$, and the proof is done.
\end{proof}

\bigskip

\begin{remark}
\label{Limit1}In the notation of the last Lemma, Schwarzschild put the
constant $\lambda=1$ by requiring that%
\[
\lim_{h\rightarrow\infty}\frac{\left[  \alpha(h)\right]  ^{2}}{(3h)^{2/3}}=1
\]
since he wanted that his solution in \textquotedblleft polar coordinates with
determinant 1\textquotedblright\ approximate the Minkowski spacetime as
$h\rightarrow\infty$. In our derivation, we are free to set $\lambda\neq0$ to
whatever we want, since this means only a change in the scale of special
coordinates $(t,h)$. We will, however, stay with the Schwarzschild choice.
\end{remark}

Now we can derive the Schwarzschild metric like we did for the Hilbert-Droste
case or in Example \ref{Uniquely1}:

\begin{proposition}
\label{Schwarzschild}Let $M=(P,(t,h),f,g,\alpha)$ be a unimodular model. Its
spacetime $\mathbf{S}_{M}$ obeys the Einstein field equation in vacuum and the
limit of Remark \emph{\ref{Limit1}} if and only if there are real numbers
$k,\mu$ such that
\[
\alpha(h)=\left(  3h+k^{3}\right)  ^{1/3}%
\]%
\[
f(h)=\frac{[\alpha(h)]^{4}}{g(h)}=1-\frac{\mu}{\alpha(h)}%
\]

\end{proposition}

\begin{proof}
The first equation follows from last Lemma. Computing the derivatives of
$\alpha$ and using the condition that%
\[
g(h)=\frac{[\alpha(h)]^{4}}{f(h)}%
\]
(recall Definition \emph{\ref{Unimodular}}) we get, by the second equation of
Proposition \emph{\ref{Einstein}},%
\[
f^{\prime}(h)=\frac{1-f(h)}{3h+\mu}%
\]
and the result follows simply by solving this equation.
\end{proof}

\begin{remark}
\label{CoordExpression}In the coordinates $(t,h)$, defined by Definition
\emph{\ref{Unimodular}} and Remark \emph{\ref{Limit}}, the metric described by
the last Proposition reads%
\begin{equation}
-\left[  1-\frac{\mu}{\alpha(h)}\right]  dt\otimes dt+\frac{1}{[\alpha
(h)]^{4}}\frac{1}{1-\mu/\alpha(h)}dh\otimes dh+[\alpha(h)]^{2}\zeta_{S^{2}}
\label{SchwarzMetric}%
\end{equation}
where $\alpha(h)=\left(  3h+k^{3}\right)  ^{1/3}$ and $\zeta_{S^{2}}$ is the
Euclidean metric of $S^{2}$. As $\alpha$ is a diffeomorphism from
$\mathbb{R}^{+}$ onto $\mathbb{R}^{+}$, we can define $(t,R)$ to be the
special coordinates of $P\subset\mathbb{R}\times\mathbb{R}^{+}$ such that
$R=\alpha\circ h$ \emph{(}recall first paragraph of Section \ref{BM}\emph{)}.
Hence, in the $(t,R)$ coordinates, the last metric reads as%
\begin{equation}
-\left(  1-\frac{\mu}{R}\right)  dt\otimes dt+\frac{1}{1-\mu/R}dR\otimes
dR+R^{2}\zeta_{S^{2}} \label{HDMetric}%
\end{equation}

\end{remark}

Thus we have the Schwarzschild metric and we known how to make a
transformation such that its coordinate expression is like that of
Hilbert-Droste. But we do not have yet the \textit{Schwarzschild solution}. As
we did in last section, we must select some submanifold $P$ of $\mathbb{R}%
\times\mathbb{R}^{+}$ to fix the topology and the spacetime manifold itself.

\begin{remark}
\label{Constant}In his original work, Schwarzschild imposed the condition that
the metric components must be smooth except in the origin of his coordinate
system. However, since our spacetime manifold is $\mathbb{R}\times
\mathbb{R}^{+}\times S^{2}$, the only way to realize that condition is by
introducing the manifold with boundary$\mathbb{\ R}\times\mathbb{[}%
0,\infty\lbrack\times S^{2}$ and extending continuously the mappings $f,g$ and
$\alpha$ from$\mathbb{\ R}^{+}$ to $\mathbb{[}0,\infty\lbrack$. Thus, as in
the boundary of$\mathbb{\ R}\times\mathbb{\{}0\}\times S^{2}$ the functions
$f$ and $g$ satisfy%
\[
\frac{1}{k^{4}}f(0)g(0)=1-\frac{\mu}{k}\text{,}%
\]
the Schwarzschild condition is in fact equivalent to%
\[
k=\mu\text{.}%
\]

\end{remark}

So now we are motivated to state the Schwarzschild solution:

\begin{definition}
[Schwarzschild solution]\label{SchwarzschildSolution}Given a real number $\mu
$, the Schwarzschild solution $S(\mu)$ is the spacetime $\mathbf{S}$ for which
there is an unimodular model $M=(P,(t,h),f,g,\alpha)$ such that $\mathbf{S=S}%
_{M}$,%
\[
P=\mathbb{R}\times\mathbb{R}^{+}%
\]%
\[
\alpha(h)=\left(  3h+\mu^{3}\right)  ^{1/3}%
\]%
\[
f(h)=\frac{[\alpha(h)]^{4}}{g(h)}=1-\frac{\mu}{\alpha(h)}%
\]

\end{definition}

\begin{problem}
Given a real number $\mu$, are the Schwarzschild $S(\mu)$ and Hilbert-Droste
$H(\mu)$ solutions equivalents?
\end{problem}

\begin{solution}
No, as they have a different topologies. The manifold which describe the
Hilbert-Droste solution is%
\[
\mathbb{R}\times(\mathbb{R}^{+}-\{\mu\})\times S^{2},
\]
while that the Schwarzschild manifold is simply%
\[
\mathbb{R}\times\mathbb{R}^{+}\times S^{2}\text{.}%
\]

\end{solution}

Because of its topology, the Hilbert-Droste solution can be sliced into two
parts, one called the exterior, whose manifold is $\mathbb{R}\times(\mu
,\infty)\times S^{2}$, and another called the interior (or the \textit{black
hole}), whose manifold is $\mathbb{R}\times(0,\mu)\times S^{2}$. As we shall
see in the next section, this topological property allow the Hilbert-Droste
manifold to be glued together with another manifold (recall Section
\ref{Exten}), constituting what is known by the Kruskal spacetime.

However, since the Schwarzschild manifold is homeomorphic to $\mathbb{R}%
\times(\mathbb{R}-\{0\})\times S^{2}$ (see, for instance, Example
\ref{HomeoR3}), we cannot find any manifold to which the Schwarzschild
manifold can be glued to, in the sense of Definition \ref{pseudo-RiemGlu}.
Even if we found in Remark \ref{CoordExpression} a coordinate transformation
such that the Schwarzschild metric acquire the same form as the
Hilbert-Droste, in the former, the metric expression holds only for $R>\mu$.

\begin{remark}
\label{Criticism}On the other hand, differently from what the author of
\emph{\cite{Abr}} did, based only on the above discussion, we cannot jump to
the conclusion that black holes do not exist as appropriated solutions of
Einstein equation. Indeed, the fact that the Einstein field equation have many
solutions with black holes seems well established, and in some cases,
according to General Relativity, black holes are unavoidable (in gravitational
collapses). What is important to keep in mind, however, is that there is no
internal mechanism in the theory to decide between the topologies of the
Schwarzschild solution and the Hilbert-Droste solution. And, in the last
analysis, the existence of black holes or the decision between the above
solutions is an experimental quest.
\end{remark}

\section{Extending the Hilbert-Droste Solution\label{Exten-HD}}

In the last Section, two descriptions of the gravitational field of a mass
point in General Relativity were discussed, the Schwarzschild and the
Hilbert-Droste solutions. However, differently from the first, the manifold of
the latter is \emph{not connected}, so it cannot qualify as a legitimate
spacetime (cf. Definition \ref{Spacetime} and, for a physical motivation, see
Remark \ref{TimeMotivation}).

In Subsection \ref{KS}, we shall extend the Hilbert-Droste manifold (cf.
Definition \ref{MaximalManifold} and Appendix \ref{Exten}) in order to obtain
a \textit{maximal spacetime}, following a procedure presented by Kruskal in
ref. \cite{Kru} and by Szekeres in ref. \cite{Sze}, both published in 1959.
Then, a historical summary (not expected to be complete) of the events which
culminated in the approach adopted by Kruskal and Szekeres (namely, the search
for new coordinate systems) is presented.

Finally, in Subsection \ref{Kasner and Fronsdal}, we discuss both from a
mathematical and a chronological standpoint an alternative extension of the
Hilbert-Droste solution by means of an embedding of that solution in a vector
manifold, an idea which began in the works of Kasner of 1929 and that was
completed by Fronsdal in 1959.

\subsection{Kruskal-Szekeres Spacetime \label{KS}}

\subsubsection{Mathematical Formalism}

In what follows, $(u,v)$ will denote the natural coordinates of $\mathbb{R}%
^{2}$.

For the sake of comparison, we start by defining the Hilbert-Droste plane and
spacetime, which are nothing more than a particular case of Definition
\ref{Plane} and a restatement of Definition \ref{Hilbert-Droste}.

\begin{definition}
\label{_Hilbertplane}Let $\mu$ be a positive real. So, the Hilbert-Droste or
simply the HD plane with mass $\mu$ is the pseudo-Riemannian manifold
$(\mathbb{R}\times(\mathbb{R}^{+}-\{\mu\}),\zeta_{H})$ such that there exists
a coordinate system $(t,r)$, which will be called the Hilbert-Droste or just
HD coordinates, for which
\[
\zeta_{H}=-\left(  1-\frac{\mu}{r}\right)  dt\otimes dt+\frac{dr\otimes
dr}{1-\mu/r}.
\]
The Hilbert-Droste or HD black hole and the normal region of the HD plane are
the pseudo-Riemannian submanifolds
\[
\mathcal{B}=(\mathbb{R}\times\left(  0,\mu\right)  ,\zeta_{H}|_{\mathbb{R}%
\times\left(  0,\mu\right)  })
\]
and%
\[
\mathcal{N}=(\mathbb{R}\times\left(  \mu,\infty\right)  ,\zeta_{H}%
|_{\mathbb{R}\times\left(  \mu,\infty\right)  })
\]
respectively.
\end{definition}

From now on, for each positive real $\mu$, the HD plane with mass $\mu$ will
be denote by $\mathcal{Q}_{H(\mu)}$, and the notation of the latter definition
will be adopted in what follows. In particular, $(t,r)$ will always denotes HD coordinates.

\begin{definition}
\label{HDSpacetime} Let $\mu$ be a positive real. Hence, the Hilbert-Droste or
HD solution is the warped product $\mathcal{Q}_{H(\mathcal{\mu)}}\times
_{r}S^{2}$, while its black hole and normal region are, respectively,
$\mathcal{B}\times_{r}S^{2}$ and $\mathcal{N}\times_{r}S^{2}$.
\end{definition}

The reader must keep in mind, however, that the last Definition do\emph{ not}
define a true spacetime, since it is not connected.

Now we start the construction of the Kruskal-Szekeres spacetime like we did in
Section \ref{BM} for the Schwarzschild case:

\begin{definition}
A Kruskal-Szekeres model, or a KS model for brevity, is an ordered list
$K=(\mu,P,f,F)$, where $\mu$ is a positive real number called the mass of $K$,
$P$ is a submanifold of $\mathbb{R}^{2}$, $f$ is a diffeomorphism from
$\mathbb{R}^{+}$ onto $[-\mu,\infty\lbrack$ $\subset\mathbb{R}$ and $F$ is a
smooth mapping from $\mathbb{R}^{+}$ into $\mathbb{R}^{+}$.
\end{definition}

\begin{definition}
\label{_Kruskalplane}The Kruskal-Szekeres or simply the KS plane associated
with a given KS model $K$ is the pseudo-Riemannian manifold $(P,\zeta_{K})$
such that%
\[
\zeta_{K}=\frac{1}{2}F(r)\left(  du\otimes dv+dv\otimes du\right)
\]
where $r=f^{-1}(uv)$.
\end{definition}

On what follows, the KS plane of a KS model $K$ will be denoted by
$\mathcal{Q}_{K}$, and its metric by $\zeta_{K}$.

\begin{remark}
\label{KruskalTime}Let $K=(\mu,P,f,F)$ be a KS model. By Problem
\emph{\ref{TimeProblem}} and by the fact that $\zeta_{K}(\partial_{u}%
-\partial_{v},\partial_{u}-\partial_{v})=-F(r)<0$, the manifold $\Pi_{K}$ is
time orientable. Then we shall call a vector $X\in T_{p}(\Pi_{K})$
future-pointing if $X$ is in the same causal cone as $\partial_{u}%
-\partial_{v}$.
\end{remark}

The next Lemma will be used in the end to prove that the Kruskal-Szekeres
spacetime is maximal. The partial derivative $\partial g/\partial x$ of a
given mapping $g$ will be denoted below by $g_{x}$.

\begin{lemma}
\label{Sectional}Let $K=(\mu,P,f,F)$ be a KS model and $S$ the sectional
curvature of $\Pi_{K}$. Thus%
\[
S(u,v)=\frac{2}{F(r)}\left[  \frac{F_{u}(r)}{F(r)}\right]  _{v}.
\]

\end{lemma}

\begin{proof}
Let $\zeta_{K}$ and $R$ be the metric and the Riemannian curvature tensor of
$\Pi_{K}$, respectively. By Definition (recall the proof of Lemma
\ref{Help1}),%
\[
S=-\frac{\zeta_{K}(R_{\partial_{u}.\partial_{v}}\partial_{u},\partial_{v}%
)}{\zeta_{K}(\partial_{u},\partial_{v})^{2}}.
\]
Let $(\Gamma_{~ij}^{k})_{(k,i,j)\in\lbrack1,2]^{3}}$ be the Christoffel
symbols and let $x^{1}=u,x^{2}=v$. By a direct computation, the only nonzero
symbols are%
\[
\Gamma_{~11}^{1}=\frac{1}{F(r)}F_{u}(r)\text{ \ \ }\Gamma_{~22}^{2}=\frac
{1}{F(r)}F_{v}(r)
\]
Therefore $R_{\partial_{u}.\partial_{v}}\partial_{u}=-D_{\partial_{v}%
},D_{\partial_{u}}\partial_{u}$ and%
\[
D_{\partial_{v}}D_{\partial_{u}}\partial_{u}=D_{\partial_{v}}\left[  \frac
{1}{F(r)}F_{u}(r)\partial_{u}\right]  =\left[  \frac{F_{u}(r)}{F(r)}\right]
_{v}\partial_{u}%
\]
Finally%
\[
S=\frac{1}{\zeta_{K}(\partial_{u},\partial_{v})}\left[  \frac{F_{u}(r)}%
{F(r)}\right]  _{v}%
\]
and the result follows by Definition \ref{_Kruskalplane}.
\end{proof}

\bigskip

In the next Definition, we shall divide the KS plane into three (not
necessarily connected) submanifolds, a procedure which will be useful in
determining in what sense the Kruskal-Szekeres spacetime contains the black
hole and the normal region of the HD solution.

\begin{definition}
\label{Regions}Let $K=(\mu,P,f,F)$ be a KS model. So, the Horizon and Regions
I and II of$\ \Pi_{K}$ are the pseudo-Riemannian submanifolds%
\begin{align*}
\mathcal{H}_{K}  &  =\{(u,v)\in P:f^{-1}(uv)=\mu\},\\
\mathcal{R}_{I}  &  =\{(u,v)\in P:f^{-1}(uv)\in\text{ }]0,\mu\lbrack\},\\
\mathcal{R}_{II}  &  =\{(u,v)\in P:f^{-1}(uv)\in\text{ }]\mu,\infty\lbrack\},
\end{align*}
with the metric inherited from $\Pi_{K}$. The positive and negative parts of
$\mathcal{R}_{I}$ and $\mathcal{R}_{II}$ are respectively%
\begin{align*}
\mathcal{R}_{I}^{+}  &  =\{(u,v)\in\mathcal{R}_{I}:u>0\},\text{ \ \ }%
\mathcal{R}_{I}^{-}=\{(u,v)\in\mathcal{R}_{I}:u<0\},\\
\mathcal{R}_{II}^{+}  &  =\{(u,v)\in\mathcal{R}_{II}:u>0\},\text{
\ \ }\mathcal{R}_{II}^{-}=\{(u,v)\in\mathcal{R}_{II}:u<0\}.
\end{align*}

\end{definition}

We shall adopt the notation of the last Definition for the rest of this section.

The following Lemma will be our main connection between the KS plane and the
HD solution.

\begin{lemma}
\label{KruskalLemma}Let $K=(\mu,P,f,F)$ be a KS model and $\mathcal{B}$ and
$\mathcal{N}$ the black hole and the normal region of $\mathcal{Q}_{H(\mu)}$,
respectively. So $\mathcal{R}_{I}^{+}$ and $\mathcal{R}_{I}^{-}$ are isometric
to $\mathcal{B}$ while $\mathcal{R}_{II}^{+}$ and $\mathcal{R}_{II}^{-}$ are
isometric to $\mathcal{N}$ if%
\begin{align*}
f(r)  &  =(r-\mu)\exp\left(  \frac{r}{\mu}\right)  ,\\
F(r)  &  =\frac{4\mu^{2}}{r}\exp\left(  -\frac{r}{\mu}\right)  .
\end{align*}

\end{lemma}

\begin{proof}
Let $\xi:\mathcal{R}_{II}^{+}\longrightarrow\mathcal{N}$ be the mapping such
that
\begin{align}
u\circ\xi(t,r)  &  =\sqrt{\left\vert r-\mu\right\vert }\exp\left(  \frac
{r+t}{2\mu}\right)  ,\label{CoordChange1}\\
v\circ\xi(t,r)  &  =\sqrt{\left\vert r-\mu\right\vert }\exp\left(  \frac
{r-t}{2\mu}\right)  , \label{CoordChange2}%
\end{align}
are surjections $\mathbb{R}\times]\mu,\infty\lbrack\longrightarrow
\mathbb{R}^{+}$. Hence, $\xi$ is a diffeomorphism.

As $du=(u\circ\xi)_{t}dt+(u\circ\xi)_{r}dr$ and $dv=(v\circ\xi)_{t}%
dt+(v\circ\xi)_{r}dr$, by Definition \ref{_Kruskalplane},%
\begin{align*}
\zeta_{K}  &  =\frac{1}{2}F(r)\left(  du\otimes dv+dv\otimes du\right)  =\\
&  F(r)(u\circ\xi)_{t}(v\circ\xi)_{t}dt\otimes dt+F(r)(u\circ\xi)_{r}%
(v\circ\xi)_{r}dr\otimes dr+\\
&  \frac{1}{2}F(r)\left[  (u\circ\xi)_{t}(v\circ\xi)_{r}+(u\circ\xi
)_{r}(v\circ\xi)_{t}\right]  (dt\otimes dr+dr\otimes dt).
\end{align*}
Then, by Definition \ref{_Hilbertplane}, we have an isometry if and only if
\begin{align}
F(r)(u\circ\xi)_{t}(v\circ\xi)_{t}  &  =-\left(  1-\frac{\mu}{r}\right)
,\label{iso1}\\
F(r)(u\circ\xi)_{r}(v\circ\xi)_{r}  &  =\frac{1}{1-\mu/r}, \label{iso2}%
\end{align}%
\begin{equation}
(u\circ\xi)_{t}(v\circ\xi)_{r}+(u\circ\xi)_{r}(v\circ\xi)_{t}=0. \label{iso3}%
\end{equation}
By computing the derivatives, Eq.(\ref{iso3}) holds trivially and
Eqs.(\ref{iso1}) and (\ref{iso2}) are satisfied only if we choose $F(r)$ as in
the Proposition. Finally,%
\[
(u\circ\xi)(v\circ\xi)=(r-\mu)\exp\left(  \frac{r}{\mu}\right)
\]
So $r=f^{-1}(uv)$ if $f(r)=(r-\mu)\exp\left(  r/\mu\right)  $.

The same reasoning holds for $\mathcal{R}_{II}^{-}$ if we redefine the
diffeomorphism $\xi:\mathcal{R}_{II}^{-}\rightarrow\mathcal{N}$ such that%
\begin{align}
u\circ\xi(t,r)  &  =-\sqrt{\left\vert r-\mu\right\vert }\exp\left(  \frac
{r+t}{2\mu}\right)  ,\\
v\circ\xi(t,r)  &  =-\sqrt{\left\vert r-\mu\right\vert }\exp\left(  \frac
{r-t}{2\mu}\right)  ,
\end{align}
Finally, to prove the isometry between $\mathcal{R}_{I}^{+}$ and
$\mathcal{R}_{I}^{-}$ and $\mathcal{B}$, just change the domains and the signs
of $u\circ\xi$ and $v\circ\xi$.
\end{proof}

\begin{remark}
\label{KruskalRemark}The reader may inquire how did we found the equations for
$u\circ\xi$ and $v\circ\xi$ in the proof of the last Proposition. Indeed, an
algebraic manipulation of Eqs. \emph{(\ref{iso1})}, \emph{(\ref{iso2})} and
\emph{(\ref{iso3}) }gives
\begin{align*}
\left\vert (u\circ\xi)_{t}\right\vert  &  =\left(  1-\frac{\mu}{r}\right)
(u\circ\xi)_{r}\text{,}\\
\left\vert (v\circ\xi)_{t}\right\vert  &  =\left(  1-\frac{\mu}{r}\right)
(v\circ\xi)_{r}\text{.}%
\end{align*}
Then, after choosing a sign for the left-hand side \emph{(}cf. the following
Remark\emph{)}, simply use separation of variables and apply some obvious
contour conditions. The details are left as an easy exercise.
\end{remark}

\begin{remark}
We can interchange the signals of the $t$--coordinate in
\emph{Eq.(\ref{CoordChange1})} and in Eq.(\ref{CoordChange2}) without
affecting the proof of the latter Lemma. However, our particular choice is the
only one in which our time-orientation of Remark \emph{\ref{KruskalTime}} is
consistent with $\partial_{t}$ being future-pointing, since%
\[
\frac{\partial}{\partial t}=\frac{\sqrt{\left\vert r-\mu\right\vert }}{2\mu
}\left[  \exp\left(  \frac{r+t}{2\mu}\right)  \frac{\partial}{\partial u}%
-\exp\left(  \frac{r-t}{2\mu}\right)  \frac{\partial}{\partial v}\right]
\]
imply that%
\[
g\left(  \partial_{t},\partial_{u}-\partial_{v}\right)  =-\frac{\sqrt
{\left\vert r-\mu\right\vert }}{2\mu}\left[  \exp\left(  \frac{r+t}{2\mu
}\right)  +\exp\left(  \frac{r-t}{2\mu}\right)  \right]  g\left(  \partial
_{u},\partial_{v}\right)  <0\text{.}%
\]

\end{remark}

In the following Exercise, the reader is invited to prove why the KS spacetime
(cf. the Definition \ref{KruskalSpacetime} below) have what some writers call
a \textquotedblleft fundamental singularity\textquotedblright\ at
\textquotedblleft$r=0$\textquotedblright.

\begin{exercise}
\label{Limit}Let $K=(\mu,P,f,F)$ be a KS model with $f$ and $F$ given by Lemma
\ref{KruskalLemma}, $S$ the sectional curvature of $\mathcal{Q}_{K}$ and
define $r=f^{-1}(uv)$. Using Lemma \emph{\ref{Sectional}}, prove by a direct
computation that $\lim_{r\rightarrow0}S=\infty$. Hint: use the following
facts: $\lim_{r\rightarrow0}u,\lim_{r\rightarrow0}v\neq0$, $F_{u}%
(r)=F^{\prime}(r)r_{u}$,%
\[
\left[  \frac{F_{u}(r)}{F(r)}\right]  _{v}=\left[  \frac{F^{\prime}(r)r_{u}%
}{F(r)}\right]  _{v}=\frac{F^{\prime}(r)}{F(r)}r_{uv}+\left[  \frac{F^{\prime
}(r)}{F(r)}\right]  _{v}r_{u}%
\]
and calculate $r_{u},r_{uv}$ implicitly by $f(r)=uv$ in Lemma\emph{
\ref{KruskalLemma}}.
\end{exercise}

Therefore, because the sectional curvature in a 2-dimensional manifold as the
KS plane becomes a real-valued mapping $\mathcal{Q}_{K}\rightarrow\mathbb{R}$,
depending exclusively of the pseudo-Riemannian structure of $\mathcal{Q}_{K}$
without any mention to a coordinate system, the \textquotedblleft
singularity\textquotedblright\ expressed in the limit $\lim_{r\rightarrow
0}S=\infty$ means that the pseudo-Riemannian structure of $\mathcal{Q}_{K}$
itself cannot be defined in the region where $f^{-1}(uv)=r=0$, justifying the
name \textquotedblleft fundamental singularity\textquotedblright$.$

Now, after the following list of Definitions and Remarks, the Kruskal-Szekeres
spacetime will be finally defined.

\begin{definition}
[KS spacetime plane]\label{KSP}Let $K=(\mu,P,f,F)$ be a KS model with%
\[
P=\{(u,v)\in\mathbb{R\times R}:uv>-\mu\}
\]
and $f$ and $F$ as defined in Lemma\emph{ \ref{KruskalLemma}}. In this case,
the KS plane $\mathcal{Q}_{K}$ is called a \emph{spacetime} plane with mass
$\mu$.
\end{definition}

Hereinafter, a KS spacetime plane with mass $\mu$ will be denoted by
$\mathcal{Q}_{K(\mathcal{\mu)}}$, in analogy with the HD plane $\mathcal{Q}%
_{H(\mu)\text{.}}$

\begin{remark}
The choice for manifold $P$ in Definition \emph{\ref{KSP}} is based on the
fact that the image of $\mathbb{R}^{+}$ under $f$ is $[-\mu,\infty\lbrack$, so
that $-\mu<f(r)=uv<\infty$.
\end{remark}

\begin{remark}
\label{HDCoordinates}By Lemma\emph{ \ref{KruskalLemma}}, there is an isometry
$\xi$ from $\mathcal{R}_{I}^{+}$ \emph{(}respect. $\mathcal{R}_{I}^{-}%
$\emph{)} onto the HD black hole $\mathcal{B}$ and another isometry, say,
$\eta$, from $\mathcal{R}_{II}^{+}$ \emph{(}respect. $\mathcal{R}_{II}^{-}%
$\emph{)} onto the HD normal region $\mathcal{N}$. Denoting again by $(t,r)$
HD coordinates, $(t,r)\circ\xi$ and $(t,r)\circ\eta$ are charts of
$\mathcal{R}_{I}^{+}$ \emph{(}respect. $\mathcal{R}_{I}^{-}$\emph{)} and
$\mathcal{R}_{II}^{+}$ \emph{(}respect. $\mathcal{R}_{II}^{-}$\emph{)}.
Together, these mappings establish a coordinate system for the manifold union
$\mathcal{R}_{I}^{+}\cup\mathcal{R}_{II}^{+}$ \emph{(}respect. $\mathcal{R}%
_{I}^{-}\cup\mathcal{R}_{II}^{-}$\emph{)}, and will be called Hilbert-Droste
or HD coordinates on the KS spacetime plane.
\end{remark}

\begin{definition}
[KS spacetime]\label{KruskalSpacetime}Let $\mu$ be a positive real number and
$f$ as in Lemma\emph{ \ref{KruskalLemma}. }The Kruskal-Szekeres spacetime, or
KS spacetime for brevity, with mass $\mathcal{\mu}$ is the warped product
$\mathcal{Q}_{K(\mu)}\times_{r}S^{2}$ where $r=f^{-1}(uv)$, and its Horizon is
the submanifold such that $uv=f(\mu)=0$.
\end{definition}

The geometric properties of the Horizon, which will be used in the next
Subsection to compare the KS spacetime with the Fronsdal embedding of the
Hilbert-Droste solution, are summarized in

\begin{lemma}
\label{Horizon}The Horizon $\mathcal{H}$ of the KS spacetime with mass $\mu$
is mapped onto $S^{2}$ by a homothety of coefficient $\mu$.
\end{lemma}

\begin{proof}
Since $uv=0$ at the Horizon, then $d\left(  u|\mathcal{H}\right)  =0$ or
$d\left(  v|\mathcal{H}\right)  =0$. In any case, the metric of $\mathcal{Q}%
_{K(\mu)}$ degenerates (cf. Definition \ref{_Kruskalplane}) and the metric of
the KS spacetime $\mathcal{Q}_{K(\mu)}\times_{r}S^{2}$ becomes just $\mu
^{2}\zeta_{S^{2}}$, where $\zeta_{S^{2}}$ is the Euclidean metric of $S^{2}$.
\end{proof}

\begin{remark}
The above Lemma is pictured as saying that the Horizon is a 2-dimensional
sphere with radius $\mu$ surrounding the \textquotedblleft fundamental
singularity\textquotedblright\ at $r=0$. That is the reason why the
traditional literature calls $r=\mu$ the \textquotedblleft Schwarzschild
radius\textquotedblright.
\end{remark}

Our first step to prove that the KS spacetime is the maximal extension of the
HD solution is to extend Lemma \ref{KruskalLemma} from the KS and HD planes to
the KS and HD \textquotedblleft spacetimes\textquotedblright.

In the proof of the next Lemma, we shall employ the natural projections
$\pi_{1}:X\times Y\longrightarrow X$ and $\pi_{2}:X\times Y\longrightarrow Y$
given by $\pi_{1}(p,q)=p$ and $\pi_{2}(p,q)=q$, for any not empty sets $X$ and
$Y$.

\begin{lemma}
Let $f$ be the mapping given by Definition \emph{\ref{KruskalLemma}} and
define again $r=f^{-1}(uv)$. Let $\mathcal{B}$ and $\mathcal{N}$ be the black
hole and normal region of $\mathcal{Q}_{H(\mu)}$ and $\mathcal{R}_{I}^{+}$,
$\mathcal{R}_{I}^{-}$, $\mathcal{R}_{II}^{+}$, $\mathcal{R}_{II}^{-}$ the
submanifolds of $\mathcal{Q}_{K(\mu)}$ as in Definition \emph{\ref{Regions}}.
So $\mathcal{R}_{I}^{+}\times_{r}S^{2}$ and $\mathcal{R}_{I}^{-}\times
_{r}S^{2}$ are isometric to $\mathcal{B}\times_{r}S^{2}$ while $\mathcal{R}%
_{II}^{+}\times_{r}S^{2}$ and $\mathcal{R}_{II}^{-}\times_{r}S^{2}$ are
isometric to $\mathcal{N}\times_{r}S^{2}$.
\end{lemma}

\begin{proof}
Let $g$, $h$, $\zeta_{I}$, $\zeta_{\mathcal{N}}$ and $\zeta_{S^{2}}$ be the
metric tensors of $\mathcal{R}_{I}^{+}\times_{r}S^{2}$, $\mathcal{N}\times
_{r}S^{2}$, $\mathcal{R}_{I}^{+}$, $\mathcal{N}$ and $S^{2}$. Letting $\xi:$
$\mathcal{R}_{I}^{+}\rightarrow\mathcal{N}$ be the isometry whose existence
was proved in Lemma\emph{ \ref{KruskalLemma}}, define $\eta:\mathcal{R}%
_{I}^{+}\times_{r}S^{2}\rightarrow\mathcal{N}\times_{r}S^{2}$ by
$\eta(p,q)=(\xi(p),q)$. \ Let%
\[
(v,w)\in T_{(p,q)}\left(  \mathcal{R}_{I}^{+}\times_{r}S^{2}\right)  \times
T_{(p,q)}\left(  \mathcal{R}_{I}^{+}\times_{r}S^{2}\right)
\]
be tangent vectors at $(p,q)$ in $\mathcal{R}_{I}^{+}\times_{r}S^{2}$, so that
$v=\pi_{1\ast}v+\pi_{2\ast}v$ for $(\pi_{1\ast}v,\pi_{2\ast}v)\in T_{p}\left(
\mathcal{R}_{I}^{+}\right)  \times T_{q}\left(  S^{2}\right)  $ and the same
for $w$ (recall Lemma \ref{Decomp} for a justification of the latter
notation). So,%
\begin{align*}
\eta^{\ast}h(v,w)  &  =\zeta_{\mathcal{N}}(\pi_{1\ast}\circ\eta_{\ast}%
,\pi_{1\ast}\circ\eta_{\ast}v)+r^{2}\zeta_{S^{2}}(\pi_{2\ast}\circ\eta_{\ast
}u,\pi_{2\ast}\circ\eta_{\ast}v)\\
&  =\zeta_{\mathcal{N}}(\xi_{\ast}\circ\pi_{1\ast}u,\xi_{\ast}\circ\pi_{1\ast
}v)+r^{2}\zeta_{S^{2}}(\pi_{2\ast}u,\pi_{2\ast}v)\\
&  =\zeta_{I}(\pi_{1\ast}u,\pi_{1\ast}v)+r^{2}\zeta_{S^{2}}(\pi_{2\ast}%
u,\pi_{2\ast}v)\\
&  =g(v,w)
\end{align*}
Hence, $\eta^{\ast}h=g$. The same argument can be repeated in order to
demonstrate the others isometries.
\end{proof}

\bigskip

Finally, we start our procedure to prove that the Kruskal-Szekeres manifold is
indeed maximal

On what follows, all geodesics $\gamma$ defined on some $I\subset\mathbb{R}$
are \textit{future-pointing}, in the sense that $\gamma^{\prime}(t)\in\sec
T_{\gamma(t)}M$ is future-pointing for all $t\in I$.

\begin{lemma}
\label{KruskalGeodesics}Let $\gamma:$ $I\subset\mathbb{R}\rightarrow
\mathcal{Q}_{K(\mu)}$ be an inextendible null geodesic. Then there exists some
$\varepsilon\in\{-1,+1\}$, a diffeomorphism $\phi$ from $J=\{s\in$
$\mathbb{R}:\varepsilon s>0\}$ onto some subset of $I$ and HD coordinates
$(t,r)$ on $\mathcal{Q}_{K(\mu)}$ such that%
\begin{align*}
r\circ\gamma\circ\phi(s)  &  =\varepsilon s,\\
t\circ\gamma\circ\phi(s)  &  =s+\varepsilon\mu\log\left\vert \mu-\varepsilon
s\right\vert ,
\end{align*}
for all $\varepsilon s\in J-\{\mu\}$.
\end{lemma}

\begin{proof}
Let $\gamma_{0}=t\circ\gamma$ and $\gamma_{1}=r\circ\gamma$ be so that
$\gamma^{\prime}(t)=\gamma_{0}^{\prime}(t)\partial_{t}+\gamma_{1}^{\prime
}(t)\partial_{r}$, where $^{\prime}$ denotes ordinary derivative and
$\{\partial_{t},\partial_{r}\}$ is the set of coordinate vector fields. By the
properties of the Levi-Civita connection and the fact that $\partial_{t}$ is a
Killing vector field, we have%
\begin{align*}
\zeta_{K}\left(  \gamma^{\prime}(t),\gamma^{\prime}(t)\right)   &  =0,\\
\zeta_{K}\left(  \gamma^{\prime}(t),\partial_{t}|_{\gamma t}\right)   &  =-K,
\end{align*}
for $K>0$ since that $\gamma^{\prime}(t)$ and $\partial_{t}|_{\gamma t}$ are
in the same causal cone. Then%
\[
\gamma_{1}^{\prime}(t)=K\text{ or }\gamma_{1}^{\prime}(t)=-K
\]%
\[
\gamma_{0}^{\prime}(t)=\frac{J}{1-\mu/\gamma_{1}(t)}%
\]
So there is a unique $\varepsilon\in\{-1,+1\}$ such that $\gamma_{1}^{\prime
}(t)=\varepsilon K$. Hence $\gamma_{1}(t)=\varepsilon\left(  Kt+C\right)  $
for $C\in\mathbb{R}$. But $\gamma$ is inextendible, so $\phi(s)=\left(
s-C\right)  /K$ must be a diffeomorphism from $J=\{s\in$ $\mathbb{R}%
:\varepsilon s>0\}$ onto some subset of $I$ (if not, $\left(  \gamma\circ
\phi(J)\right)  \cap I\neq\varnothing$ and $\gamma$ would not be
inextendible). Thus $\gamma_{1}\circ\phi(s)=\varepsilon s$ and%
\[
\left(  \gamma_{2}\circ\phi\right)  ^{\prime}(s)=\frac{1}{1-\mu/\left(
\varepsilon s\right)  }%
\]
Hence, there is $D\in\mathbb{R}$ for which $\gamma_{2}\circ\phi
(s)=D+s+\varepsilon\mu\log\left\vert \mu-\varepsilon s\right\vert $, and
result follows by the fact that $(t,r)\rightarrow(t+D,r)$ is an isometry.
\end{proof}

\bigskip

Finally:

\begin{proposition}
\label{KruskalMaximal}The KS spacetime is maximal.
\end{proposition}

\begin{proof}
As $S^{2}$ is connected and compact, we just need to prove that the KS
spacetime plane $\mathcal{Q}_{K(\mu)}$ is maximal. Let $f$ be as in Lemma
\ref{KruskalLemma} and, as usual, define $r=f^{-1}(uv)$. Let $M$ be a
spacetime in which $\mathcal{Q}_{K(\mu)}$ is a submanifold, $\gamma$ an
inextendible null geodesic in $M$ and $J\subset I$ the largest subset such
that $\gamma(J)\subset\mathcal{Q}_{K(\mu)}$. Assume that $J$ is not empty. By
Lemma \ref{KruskalGeodesics}, as $\gamma$ is inextendible, $k\in I-J$ implies
$r\circ\gamma(k)=0$ or $r\circ\gamma(k)=\mu$. But by Exercise \ref{Limit}, in
the first case: $\lim_{t\rightarrow k}S\circ\gamma(t)=\infty$. So
$r\circ\gamma(k)=\mu$. However, this imply $\gamma(k)\in H\subset
$.$\mathcal{Q}_{K(\mu)}$. So $I-J=\varnothing$ and the result follows by Lemma
\ref{Maximal}.
\end{proof}

\bigskip

In the following Problem, the reader is invited to see why the
\textquotedblleft interior submanifold\textquotedblright\ $\mathcal{R}_{I}%
^{+}$ of the HD solution is called a \emph{black hole}, while $\mathcal{R}%
_{I}^{-}$ is usually know as a \emph{white hole}.

\begin{problem}
Let $\gamma:$ $I\subset\mathbb{R}\rightarrow\mathcal{Q}_{K(\mu)}$ be a
\emph{(}future-pointing\emph{) }null geodesic. Suppose there exists some $k\in
I$ such that $\gamma(k)\in\mathcal{R}_{I}^{+}$ (respect. $\mathcal{R}_{I}^{-}%
$). Hence $\gamma|\left(  I\cap\lbrack k,\infty)\right)  \subset
\mathcal{R}_{I}^{+}$ (respect. $\exists s\in I:\gamma|\left(  I\cap\lbrack
s,\infty)\right)  \subset\mathcal{R}_{I}^{-}$).
\end{problem}

\begin{solution}
In $\mathcal{R}_{I}^{+}$, $dr$ is a timelike vector, and so is
$\operatorname{grad}r$. Hence, as $v<0$ and $u>0$ in $\mathcal{R}_{I}^{+}$,%
\[
\zeta_{K}(\operatorname{grad}r,\partial_{u}-\partial_{v})=\frac{\partial
r}{\partial u}-\frac{\partial r}{\partial v}=\frac{v-u}{r}\exp\left(  \frac
{r}{\mu}\right)  <0
\]
So $\operatorname{grad}r$ is future-pointing. Therefore%
\[
\left(  r\circ\gamma\right)  ^{\prime}(t)=dr\left[  \gamma^{\prime}(t)\right]
=\zeta_{K}(\operatorname{grad}r,\gamma^{\prime}(t))<0
\]
Now, in $\mathcal{R}_{I}^{-}$, $v>0$ and $u<0$ so that $\zeta_{K}%
(\operatorname{grad}r,\partial_{u}-\partial_{v})>0$. Therefore,
$\operatorname{grad}r$ is past-pointing in $\mathcal{R}_{I}^{-}$, so that
$\left(  r\circ\gamma\right)  ^{\prime}(t)>0$.
\end{solution}

At this point, the reader may consult the Appendix \ref{ERB} to see our brief
discussion concerning some exotic topological objects associated with the KS
spacetime. There, we shall prove the important fact that there exists no
geodesic which starts on the region $\mathcal{R}_{II}^{+}$ and goes to
$\mathcal{R}_{II}^{-}$ and vice-versa, or that starts on $\mathcal{R}_{I}^{+}$
and goes to $\mathcal{R}_{I}^{-}$.

This result has been pictured by some writers as saying that $\mathcal{R}%
_{II}^{+}$ and $\mathcal{R}_{II}^{-}$ ($\mathcal{R}_{I}^{+}$ and
$\mathcal{R}_{I}^{-}$, respectively) belongs to distinct \textquotedblleft
universes\textquotedblright\ which are connected by the horizon $\mathcal{H}$
of the KS spacetime. Some authors, e.g., Kruskal himself \cite{Kru}, have
compared the latter object to a \textquotedblleft wormhole\textquotedblright,
in the Misner and Wheeler sense, or to a \textquotedblleft
bridge\textquotedblright, in the Einstein-Rosen sense. However, because one
cannot travel through it without falling in the \textquotedblleft fundamental
singularity\textquotedblright, the more accurate term \textit{horizon} has
been adopted in the literature.

Because it is impossible even in principle to verify the existence of the
another \textquotedblleft universe\textquotedblright\ represented by a
manifold union like $\mathcal{R}_{I}^{-}\cup\mathcal{R}_{II}^{-}$, we may
ignore it and \textit{truncate} the KS spacetime by removing $\mathcal{R}%
_{I}^{-}\cup\mathcal{R}_{II}^{-}$ from the KS spacetime manifold. This is done
in many textbooks, cf. for instance ref. \cite{On1} and ref. \cite{Sac}.
However, it is important to keep in mind that the \textit{maximal} manifold is
the KS spacetime \textit{with the two \textquotedblleft
universes\textquotedblright}, so that, for our geometrical purposes, we shall
let Definition \ref{KruskalSpacetime} as it is without truncating the spacetime.

\subsubsection{Historical Overview: Hunting for New Coordinates\label{HO}}

The maximal extension of the Hilbert-Droste solution is an important chapter
in the history of the General Relativity and could very well deserve an entire
work dedicated to it. In the next pages, we give a brief summary of the
research programme established from the decade of 1920 to the end of 1960,
which is completed in the following Section with a detailed discussion of the
Fronsdal embedding.

In this section, let $Q_{K}$ be the Kruskal-Szekeres spacetime plane and
denote by $(t,r)$ the Hilbert-Droste chart on $Q_{K}$.

The idea of finding a coordinate system which could remove the singularity of
the Hilbert-Droste metric at $r=\mu$ was pioneered by the French mathematician
and former Prime Minister of Third Republic P. Painlev\'{e}, in 1921
\cite{Pain}. A year later, the same coordinates proposed by Painlev\'{e} were
restated by the Swedish ophthalmologist A. Gullstrand \cite{Gul}, who won the
\textit{Nobel Prize in Physiology or Medicine} for his works on the eye
optics, using methods of applied mathematics \cite{Nob}.

They introduced a new coordinate, say, $z$, for which the metric of $Q_{K}$
could be rewritten as
\[
\zeta_{K}=-\left(  1-\frac{\mu}{r}\right)  dz\otimes dz+\left(  dz\otimes
dr+dr\otimes dz\right)
\]
Using the same arguments of Lemma \ref{KruskalLemma} and Remark
\ref{KruskalRemark}, one can find the transformation identity of $z$ in terms
of the Hilbert-Droste coordinates, which gives%
\[
z=t+r+\mu\log\left\vert \frac{r}{\mu}-1\right\vert .
\]
For more details on this coordinates, see \cite{PG}.

\begin{remark}
Ironically, Gullstrand was the member of the Nobel Committee for Physics who
argue against Einstein receiving the Nobel Prize for his works in the
Relativity Theory. Consult, for instance, \emph{\cite{Rav}}.
\end{remark}

The spokesman for Relativity Theory in the English-speaking world during the
first great war, A. Eddington, found another coordinate system in 1924 which
also removed the $r=\mu$ singularity \cite{Edi}. He introduced a new
\textquotedblleft time\textquotedblright\ coordinate $\tilde{t}$ such that%
\[
\zeta_{K}=-\left(  1-\frac{\mu}{r}\right)  d\tilde{t}\otimes d\tilde{t}%
+\frac{\mu}{r}\left(  d\tilde{t}\otimes dr+dr\otimes d\tilde{t}\right)
+\left(  1+\frac{\mu}{r}\right)  dr\otimes dr
\]
and, as in Remark \ref{KruskalRemark}, one can find that%
\[
\tilde{t}=t-\mu\log\left\vert r-\mu\right\vert .
\]
The $(\tilde{t},r)$ coordinates were finally used by D. Finkelstein in 1958
\cite{Fin} in order to find the maximal extension of the Hilbert-Droste
\textquotedblleft spacetime\textquotedblright. However, in the decade of 1920,
Eddington's motivation was very different from Finkelstein's. The former was
concerned with the fact that the Whitehead gravitational theory gave the same
predictions as the Relativity Theory for the solar system. Eddington, using
the last coordinate expression for the metric, prove that, indeed, the
Hilbert-Droste metric was a solution of both gravitational theories.

Another step was given in 1933 by the proposer of the expansion of the
universe and of the \textquotedblleft primeval atom\textquotedblright\ theory,
G. Lema\^{\i}tre. He introduced a set of coordinates $(\tau,\rho)$ for which%
\begin{align*}
\zeta_{K}  &  =-d\tau\otimes d\tau+\frac{\mu}{r}d\rho\otimes d\rho,\\
r  &  =\left[  \frac{2}{3}\sqrt{\mu}\left(  \rho-\tau\right)  \right]  ^{2/3}%
\end{align*}
and, as he wrote in his paper \cite{Lem},

\begin{quotation}
La singularit\'{e} du champ de Schwarzschild est donc une singularit\'{e}
fictive, analogue \`{a} celle qui se pr\'{e}sentait \`{a} l'horizon du centre
dans la forme originale de l'univers de De Sitter.
\end{quotation}

\begin{remark}
Some people argues that C. Lanczos was in fact the first to institute the idea
of a \textquotedblleft fictitious\textquotedblright\ singularity. In
opposition to Lema\^{\i}tre, he did this by introducing a singularity in a
solution which were known as regular. See \emph{\cite{Eis}}.
\end{remark}

The maximal extension of the Hilbert-Droste \textquotedblleft
spacetime\textquotedblright\ became a theme of mainstream research more or
less in 1949, with the publication of a \textit{Letter to the Editor} in the
\textit{Nature} magazine \cite{Syn} by J. Synge, an Irish mathematician and
physicist who made contributions to differential geometry (Synge's theorem)
and to theoretical physics. He summarized the state of affairs as follows:

\begin{quotation}
\textquotedblleft The usual exterior Schwarzschild line element shows an
obvious singularity for a certain value of $r$, say $r=a$. Since $a$ is in
every known case much smaller than the radius of the spherical body producing
the field, the existence of the singularity appears to be of little interest
to astronomers. At the other end of the scale, a discussion of the
gravitational field of an ultimate particle, without reference to
electromagnetism or quantum theory, might appear equally devoid of physical
meaning.\newline

Nevertheless, it does not seem right to leave the theory of the gravitational
field of a particle (I mean a point-mass) uncompleted merely because there is
no direct physical application. The existence of the singularity at $r=a$ is
strange and demands investigation. Investigation shows that there is no
singularity at $r=a$; the apparent singularity in the Schwarzschild line
element is due to the coordinates employed and may be removed by a
transformation, so that there remains no singularity except at $r=0$%
\textquotedblright.
\end{quotation}

A year after, Synge published a paper \cite{Syn1} where he also proposed some
new coordinates $(u,v)$ for which the metric could be given by%
\[
\zeta_{K}=-\left(  1+v^{2}G\right)  du\otimes du+\frac{1}{2}uv\left(
du\otimes dv+dv\otimes du\right)  +\left(  1-u^{2}G\right)  dv\otimes dv,
\]
with%
\begin{align*}
G  &  =\frac{1}{u^{2}-v^{2}}\left(  1-\frac{4\mu^{2}\tanh^{2}\xi}{u^{2}-v^{2}%
}\right)  \text{ if }u^{2}-v^{2}\geq0,\\
G  &  =\frac{1}{u^{2}-v^{2}}\left(  1+\frac{4\mu^{2}\tanh^{2}\eta}{u^{2}%
-v^{2}}\right)  \text{ if }u^{2}-v^{2}<0,
\end{align*}
where%
\begin{align*}
r  &  =\mu\cosh^{2}\xi\text{ if }u^{2}-v^{2}\geq0,\\
r  &  =\mu\cos^{2}\eta\text{ if }u^{2}-v^{2}<0.
\end{align*}
Therefore, he preceded to study the motion of geodesics falling into the
\textquotedblleft interior region $r<\mu$\textquotedblright\ and became the
first one to explore the physics of the Hilbert-Droste black hole.

\begin{remark}
It is notable that Synge used a coordinate system much more complicated than
that used by Eddington and Finkelstein or even the one used by Lema\^{\i}tre
\emph{(}which he knew very well\emph{)}. However, it is not certain if Synge
was aware of the coordinates employed by Painlev\'{e} and Gullstrand which,
even being the first one to appear, are relatively simpler. In fact, there are
some authors defending that the latter coordinates are better employed even in
a pedagogical context. See \emph{\cite{PG}}.
\end{remark}

After the endeavour of Synge and then by Finkelstein, the American
mathematical physicist M. Kruskal and the Hungarian-Australian mathematician
G. Szekeres employed a comparatively simpler coordinate system to remove the
singularity at $r=\mu$. They required that the coordinates $(x,y)$ are to be
such that%
\[
\zeta_{K}=-F(r)dx\otimes dx+F(r)dy\otimes dy.
\]
The reader may then take as an Exercise to prove, using the techniques of
Lemma \ref{KruskalLemma}, that%
\begin{align*}
\left[  F(r)\right]  ^{2}  &  =\frac{16\mu}{r}\exp\left(  -\frac{r}{\mu
}\right)  ,\\
r  &  =f^{-1}(x^{2}-y^{2}),
\end{align*}
where%
\[
f(r)=\left[  \left(  \frac{r}{\mu}\right)  -1\right]  \exp\left(  \frac{r}%
{\mu}\right)
\]
and the coordinate transformations are given by%
\begin{align}
x  &  =\sqrt{\left(  \frac{r}{\mu}\right)  -1}\exp\left(  \frac{r}{2\mu
}\right)  \cosh\left(  \frac{t}{2\mu}\right)  ,\label{CoordChange3}\\
y  &  =\sqrt{\left(  \frac{r}{\mu}\right)  -1}\exp\left(  \frac{r}{2\mu
}\right)  \sinh\left(  \frac{t}{2\mu}\right)  . \label{CoordChange4}%
\end{align}
The resemblance between these coordinates and that used in the last section
(recall, for instance, Lemma \ref{KruskalLemma}) is not a coincidence. The
chart used in Section \ref{KS} is just a reformulation of that used by Kruskal
and Szekeres, and appeared already in the first edition of \cite{Haw}.

\subsection{Kasner-Fronsdal Embedding\label{Kasner and Fronsdal}}

There is another aspect of the history of the maximal extension which started
already in 1921 with the American mathematician E. Kasner, pupil of F. Klein
and D. Hilbert. Kasner is celebrated by his works in Differential Geometry,
General Relativity and even the popularization of the term \textit{googol} to
refer to the number $10^{100}$, which would later inspire the name of the
famous Internet search engine.

Kasner started by proving in \cite{Kas} that it is impossible to find an
imbedding in a flat 5-dimensional manifold of a (non-Euclidean) spacetime
which satisfies Einstein field equation in vacuum. (His proof is remarkably
simple and elementary). And in a following paper \cite{Kas1} (published with
the former in the same Volume of the American Journal of Mathematics), Kasner
demonstrated how the Schwarzschild solution can be \textquotedblleft
embedded\textquotedblright\ (with a topological defect, cf. the paragraph
following Lemma \ref{KasnerLemma}) in a flat 6-dimensional manifold.

His latter work will be described in details below.

\begin{remark}
In \emph{\cite{Fron}}, C. Fronsdal observed correctly that the embedding
proposed by Kasner was not valid for the whole Hilbert-Droste
\textquotedblleft spacetime\textquotedblright. Fronsdal argues that he
imbedded only the \textquotedblleft exterior solution $r>\mu$%
\textquotedblright, with an additional topological modification.
\end{remark}

First, we define the manifold in which we pretend to imbed the Schwarzschild spacetime:

\begin{definition}
\label{Kasner}The Kasner manifold is the 6-dimensional pseudo-Riemannian
vector manifold $V$ for which there is a natural isomorphism $(u_{1}%
,u_{2},...,u_{6})$ from $V$ onto $\mathbb{R}^{6}$, which will be called the
natural coordinates of $V$, such that the metric $\zeta_{V}$ can be written as%
\[
\zeta_{V}=-\sum_{i\in\lbrack1,2]}du_{i}\otimes du_{i}+\sum_{i\in\lbrack
3,6]}du_{i}\otimes du_{i}%
\]

\end{definition}

In order to simplify his work, Kasner introduced a new coordinate system in
the Schwarzschild spacetime. On what follows, the Euclidean metric of
$\mathbb{R}^{3}-\{0\}$ will be denoted by $\zeta_{\mathbb{R}^{3}}$.

\begin{definition}
\label{Kasner Chart}Let $S=S(\mu)$ be a Schwarzschild solution \emph{(}cf.
Definition \emph{\ref{SchwarzschildSolution})} and $(t,R)$ a chart of the
Schwarzschild plane $\Pi$ of $S$ such that the metric $\zeta_{S}$ of $S$ is
given by%
\[
\zeta_{S}=-\left(  1-\frac{\mu}{R}\right)  dt\otimes dt+\frac{1}{1-\mu
/R}dR\otimes dR+\zeta_{\mathbb{R}^{3}}%
\]
for $R>\mu$ \emph{(}cf. Remarks \emph{\ref{CoordExpression}} and
\emph{\ref{Constant})}. Let $h\rightarrow\phi(h)=\sqrt{4\mu^{2}(-1+h/\mu)}$ be
a diffeomorphism from $\left]  \mu,\infty\right[  \subset\mathbb{R}$ onto
$\mathbb{R}^{+}$. So, the Kasner chart is the coordinate system $(t,H)$ on
$\Pi$ for which $H=\phi\circ R$.
\end{definition}

\begin{exercise}
Assuming the notation of the last Definition, prove that the metric $\zeta
_{S}$ of $S$ is, in the Kasner chart, given by%
\[
\zeta_{S}=-\frac{H^{2}}{H^{2}+4\mu^{2}}dt\otimes dt-dH\otimes dH+\zeta
_{\mathbb{R}^{3}}\text{.}%
\]

\end{exercise}

The next Lemma will finally give the imbedding condition. As usual, the
partial derivative $\partial g/\partial x$ of a given mapping $g$ will be
denoted by $g_{x}$.

\begin{lemma}
\label{KasnerLemma}Let $V$ and $S=S(\mu)$ be the Kasner manifold and the
Schwarzschild solution with mass $\mu$, and denote by $\zeta_{V}$ and
$\zeta_{S}$ their respective metrics. Let $(u_{1},u_{2},...,u_{6})$ be the
natural coordinates of $V$, $(t,H)$ the Kasner chart of the Schwarzschild
plane of $S$ and $(t,x)$ a coordinate system of $S$ related to $(t,H)$ by%
\[
H=\left\Vert x\right\Vert =\sqrt{x_{1}^{2}+x_{2}^{2}+x_{3}^{2}}\text{,
\ \ }x=(x_{1},x_{2},x_{3})\text{.}%
\]
So, a mapping $\xi$ from $S$ onto some submanifold $\Bbbk\subset V$ is an
isometry if%
\begin{align*}
u_{1}\circ\xi(t,x)  &  =\frac{H}{\sqrt{H^{2}+4\mu^{2}}}\sin t,\\
u_{2}\circ\xi(t,x)  &  =\frac{H}{\sqrt{H^{2}+4\mu^{2}}}\cos t,
\end{align*}%
\[
u_{3}\circ\xi(t,x)=\int_{0}^{H}\sqrt{1+\frac{16\mu^{4}}{\left(  H^{2}+4\mu
^{2}\right)  ^{3}}}dH
\]
and
\[
u_{4}\circ\xi(t,x)=x_{1}\text{, \ \ }u_{5}\circ\xi(t,x)=x_{2}\text{,
\ \ }u_{6}\circ\xi(t,x)=x_{3}\text{.}%
\]
In this case, $\Bbbk$ will be called the Kasner imbedding.
\end{lemma}

\begin{proof}
Applying the condition $\xi^{\ast}(\zeta_{V}|\Bbbk)=\zeta_{S}$ to the
coordinate vectors $\partial_{t}$, $\partial_{H}$, we derive that%
\[
(u_{1}\circ\xi)_{r}^{2}+(u_{2}\circ\xi)_{r}^{2}-(u_{3}\circ\xi)_{r}^{2}=1,
\]%
\[
(u_{1}\circ\xi)_{t}^{2}+(u_{2}\circ\xi)_{t}^{2}-(u_{3}\circ\xi)_{t}^{2}%
=\frac{H^{2}}{H^{2}+4\mu^{2}},
\]%
\[
(u_{1}\circ\xi)_{t}(u_{1}\circ\xi)_{r}+(u_{2}\circ\xi)_{t}(u_{2}\circ\xi
)_{r}-(u_{3}\circ\xi)_{t}(u_{3}\circ\xi)_{r}=0.
\]
We set $u_{4}\circ\xi(t,x)=x_{1}$, $u_{5}\circ\xi(t,x)=x_{2}$, $u_{6}\circ
\xi(t,x)=x_{3}$. Now, let $F,W$ be mappings from $\mathbb{R}^{+}$ into
$\mathbb{R}$ and let $G,Z$ be mappings from $\mathbb{R}$ into $\mathbb{R}$,
such that $u_{1}\circ\xi(t,x)=G(t)F(H)$, $u_{2}\circ\xi(t,x)=Z(t)F(H)$ and
$u_{3}\circ\xi(t,x)=W(H)$. Then from the last equation,%
\[
\left[  G(t)^{2}+Z(t)^{2}\right]  ^{\prime}=0,
\]
which holds if we choose $G(t)=\sin t$ and $Z(t)=\cos t$. But by the second,%
\[
\left[  F(H)\right]  ^{2}\left[  G^{\prime}(t)^{2}+Z^{\prime}(t)^{2}\right]
=\frac{H^{2}}{H^{2}+4\mu^{2}},
\]
hence $F(H)=H/\sqrt{H^{2}+4\mu^{2}}$. Finally,%
\[
\left[  F^{\prime}(H)\right]  ^{2}=1+\left[  W^{\prime}(H)\right]  ^{2}%
\]
follows from the first equation.
\end{proof}

\bigskip

Assuming the notation of the last Lemma, $\Bbbk$ is then an imbedding of the
Schwarzschild spacetime into the Kasner manifold. However, the points $(t,x)$
and $(t+2\pi,x)$ are identified in $\Bbbk$, and therefore, such an imbedding
have the\ exotic topology of $S^{1}\times\mathbb{R}^{+}\times S^{2}$,
providing an example of what we may call a \emph{naive time machine}.

\begin{remark}
Another naive time machine can be \textquotedblleft
constructed\textquotedblright\ as follows. Let $\mathcal{M}$ be the Minkowski
spacetime and let $\eqsim$ be the equivalence relation in $\mathcal{M}$ such
that $(t,x,y,z)\eqsim(t+2\pi,x,y,z)$ in a Lorentz system. So $\mathcal{M}%
/\eqsim$, homeomorphic to $S^{1}\times\mathbb{R}^{3}$, is the simpler case of
a naive time machine!
\end{remark}

\begin{remark}
We have again an illustration of the topological arbitrariness that exists in
General Relativity. In the present case, the choice between the Schwarzschild
manifold with the topology of $\mathbb{R}\times\mathbb{R}^{+}\times S^{2}$,
and not with that of $S^{1}\times\mathbb{R}^{+}\times S^{2}$, is a matter of
experimentation: we know that there is no such a \textit{time machine} in our
solar system. \emph{(}Compare this with the situation discussed in Remark
\emph{\ref{Criticism})}.
\end{remark}

As we would expect from our earlier discussion, the Kasner imbedding do not
provide a completion for the Schwarzschild solution, since the diffeomorphism
$\left]  \mu,\infty\right[  \overset{\phi}{\longrightarrow}\mathbb{R}^{+}$ of
Definition \ref{Kasner Chart} cannot be extended to $\left]  0,\mu\right[  $.

Lastly, in order to imbed the Hilbert-Droste disconnected \textquotedblleft
spacetime\textquotedblright, the American theoretical physicist C. Fronsdal
completed the maximal extension programme in 1959 by modifying the Kasner
manifold. Historically, one of his motivations was to remove the \textit{naive
time machine}. As he wrote in \cite{Fron} (where $Z_{1}=$ $u_{1}\circ\xi$ and
$Z_{2}=u_{2}\circ\xi$)

\begin{quotation}
Another shortcoming (...) is that $Z_{1}$ and $Z_{2}$ are periodic functions
of $t$, so that the embedding identifies distinct points of the original
manifold. This suggests replacing the trigonometric functions by hyperbolic functions.
\end{quotation}

To make a parallel with the Kasner work, we present the Fronsdal construction
by modifying the Definition \ref{Kasner}.

\begin{definition}
\label{FronsdalManifold}The Fronsdal manifold is the $6$-dimensional
pseudo-Riemannian vector manifold $U$ for which there is a natural isomorphism
$(u_{1},u_{2},...,u_{6})$ from $U$ onto $\mathbb{R}^{6}$, which will be called
the natural coordinates of $U$, such that the metric $\zeta_{U}$ can be
written as%
\[
\zeta_{U}=-du_{1}\otimes du_{1}+\sum_{i\in\lbrack2,6]}du_{i}\otimes du_{i}.
\]

\end{definition}

As we shall see in the proof of the next Lemma, the change of sign from the
metric of Definition \ref{Kasner} to the above is necessary because of the
hyperbolic mappings adopted Fronsdal.

\begin{lemma}
\label{Fronsdal}Let $U$ and $H=H(\mu)$ be the Fronsdal manifold and the
Hilbert-Droste solution with mass $\mu$, and denote by $\zeta_{U}$ and
$\zeta_{H}$ their respective metrics. Let $\mathcal{B}\times_{r}S^{2}$ and
$\mathcal{N}\times_{r}S^{2}$ be the black hole and the normal region of $H$,
$(u_{1},u_{2},...,u_{6})$ the natural coordinates of $U$, $(t,r)$ the
Hilbert-Droste coordinates of the plane of $H$ and $(t,x)$ the coordinate
system of $H$ related to $(t,r)$ by%
\[
r=\left\Vert x\right\Vert =\sqrt{x_{1}^{2}+x_{2}^{2}+x_{3}^{2}}\text{,
\ \ }x=(x_{1},x_{2},x_{3})\text{.}%
\]
So, the mappings $B\overset{\eta}{\longrightarrow}\Bbbk_{B}$ and
$N\overset{\xi}{\longrightarrow}\Bbbk_{N}$ where $\Bbbk_{B},\Bbbk_{N}\subset
U$ are submanifolds are isometries if%
\begin{align*}
u_{1}\circ\xi(t,x)  &  =2\mu\sqrt{1-\frac{\mu}{r}}\sinh\left(  \frac{t}{2\mu
}\right)  \text{, }u_{1}\circ\eta(t,x)=2\mu\sqrt{\frac{\mu}{r}-1}\cosh\left(
\frac{t}{2\mu}\right)  ,\\
u_{2}\circ\xi(t,x)  &  =2\mu\sqrt{1-\frac{\mu}{r}}\cosh\left(  \frac{t}{2\mu
}\right)  \text{, }u_{2}\circ\eta(t,x)=2\mu\sqrt{\frac{\mu}{r}-1}\sinh\left(
\frac{t}{2\mu}\right)  ,
\end{align*}%
\[
u_{3}\circ\xi(t,x)=u_{3}\circ\eta(t,x)=\int_{0}^{r}\sqrt{\frac{(h+\mu
)(h^{2}+\mu^{2})}{h^{3}}}dh
\]
and as before,
\[
u_{4}\circ\xi(t,x)=u_{4}\circ\eta(t,x)=x_{1}\text{, \ \ }u_{5}\circ
\xi(t,x)=u_{5}\circ\eta(t,x)=x_{2}\text{,}%
\]%
\[
u_{6}\circ\xi(t,x)=u_{6}\circ\eta(t,x)=x_{3}\text{.}%
\]
In this case, $(\Bbbk_{B},\Bbbk_{N})$ will be called the Fronsdal structure.
\end{lemma}

\begin{proof}
By the condition that $\xi^{\ast}(\zeta_{U}|\Bbbk_{N})=\zeta_{H}$, we obtain
that%
\begin{align*}
-(u_{1}\circ\xi)_{r}^{2}+(u_{2}\circ\xi)_{r}^{2}+(u_{3}\circ\xi)_{r}^{2}  &
=\frac{1}{1-\mu/r},\\
-(u_{1}\circ\xi)_{t}^{2}+(u_{2}\circ\xi)_{t}^{2}+(u_{3}\circ\xi)_{t}^{2}  &
=-\left(  1-\frac{\mu}{r}\right)  ,
\end{align*}%
\[
(u_{1}\circ\xi)_{t}(u_{1}\circ\xi)_{r}+(u_{2}\circ\xi)_{t}(u_{2}\circ\xi
)_{r}-(u_{3}\circ\xi)_{t}(u_{3}\circ\xi)_{r}=0.
\]
We set $u_{4}\circ\xi(t,x)=x_{1}$, $u_{5}\circ\xi(t,x)=x_{2}$, $u_{6}\circ
\xi(t,x)=x_{3}$. Let $F,W$ be mappings from $\mathbb{R}^{+}$ into $\mathbb{R}$
and let $G,Z$ be mappings from $\mathbb{R}$ into $\mathbb{R}$, such that
$u_{1}\circ\xi(t,x)=G(t)F(r)$, $u_{2}\circ\xi(t,x)=Z(t)F(r)$ and $u_{3}%
\circ\xi(t,x)=W(r)$. From the above equations,%
\[
\left[  G(t)^{2}-Z(t)^{2}\right]  ^{\prime}=0,
\]%
\[
\left[  F(r)\right]  ^{2}\left[  G^{\prime}(t)^{2}-Z^{\prime}(t)^{2}\right]
=1-\frac{\mu}{r},
\]%
\[
\left[  F^{\prime}(r)\right]  \left[  Z(t)^{2}-G(t)^{2}\right]  =\frac
{1}{1-\mu/r}%
\]
and the result follows simply by taking $G(t)=\sinh\left(  \frac{t}{2\mu
}\right)  $ and $Z(t)=\cosh\left(  \frac{t}{2\mu}\right)  $. The proof is same
for the black hole.
\end{proof}

\bigskip

Using the last coordinate equations for $\xi$ and $\eta$, one can prove very
easily the following Lemma:

\begin{lemma}
\label{FronsdalImb}Assume the hypothesis of Lemma \emph{\ref{Fronsdal}}. Given
a positive real number $\mu$, let $F(\mu)$ be the hypersurface of the Fronsdal
manifold $U$ such that $p\in F(\mu)$ if and only if there is some real $r>0$
for which%
\[
u_{2}(p)^{2}-u_{1}(p)^{2}=4\mu^{2}\left(  1-\frac{\mu}{r}\right)  ,
\]%
\[
u_{3}(p)=\int_{0}^{r}\sqrt{\frac{(h+\mu)(h^{2}+\mu^{2})}{h^{3}}}dh,
\]%
\[
u_{4}(p)^{2}+u_{5}(p)^{2}+u_{6}(p)^{2}=r^{2}.
\]
So the Fronsdal structure $(\Bbbk_{N},\Bbbk_{B})$ is such that $\Bbbk_{N}%
\cup\Bbbk_{B}\subset F(\mu)$.
\end{lemma}

Because of the latter Lemma, $F(\mu)$ will be called the Fronsdal hypersurface
with mass $\mu$, while $\Bbbk_{B}$ and $\Bbbk_{N}$ can be identified as the
black hole and normal region belonging to the Fronsdal hypersurface.

\begin{exercise}
In the notation of Lemma\emph{ \ref{FronsdalImb}}, prove that the mapping
$\lambda$ from $F(\mu)$ onto $F(\mu)$ such that $u_{1}\circ\lambda
(p)=-u_{1}(p)$, $u_{2}\circ\lambda(p)=-u_{2}(p)$ and $u_{i}\circ
\lambda(p)=u_{i}(p)$ for all $i\in\lbrack3,6]\subset\mathbb{N}$ is an isometry.
\end{exercise}

Thus, by the last Exercise, there are\emph{ two} copies of the Hilbert-Droste
manifold in the Fronsdal hypersurface, as we would expect from the existence
of two \textquotedblleft universes\textquotedblright\ belonging to the
Kruskal-Szekeres spacetime.

We are finally motivated to define the maximal extension of the Hilbert-Droste
solution from the Fronsdal-Kasner approach:

\begin{definition}
[Fronsdal spacetime]The Fronsdal spacetime with mass $\mu$ is simply the
Fronsdal hypersurface $F(\mu)$ with the metric induced from the vector
manifold $U$ (cf. Definition \emph{\ref{FronsdalManifold}}).
\end{definition}

Following the reasoning presented in the last section, we could study the
geodesics of the Fronsdal spacetime in order to prove that it is indeed
maximal. However, it will be more instructive if we explore the relation
between the Kruskal-Szekeres spacetime and the former.

In fact, because of Proposition \ref{KruskalMaximal} and Lemma \ref{Fronsdal},
we only need to show that the Horizon that connects the black hole and the
normal region in the Fronsdal hypersurface is equivalent to that of the
Kruskal-Szekeres spacetime.

\begin{definition}
[Horizon]Let $U$ be the Fronsdal manifold with natural coordinates
$(u_{1},u_{2},...,u_{6})$, $F=F\mathbf{(}\mu)$ the Fronsdal spacetime with
mass $\mu$ and $r$ the mapping from $F$ into $\mathbb{R}$ such that
\[
r(p)=\sqrt{u_{4}(p)^{2}+u_{5}(p)^{2}+u_{6}(p)^{2}}\text{.}%
\]
The Horizon $\mathcal{H}$ of $F$ is the submanifold for which $p\in
\mathcal{H}$ if and only if $r(p)=\mu$.
\end{definition}

\begin{lemma}
\label{Horizon 2}Let $\mathcal{H}_{1}$ and $\mathcal{H}_{2}$ be the Horizons
of the Kruskal-Szekeres (cf. Definition \emph{\ref{KruskalSpacetime}}) and
Fronsdal spacetimes with mass $\mu$. Then $\mathcal{H}_{1}$ and $\mathcal{H}%
_{2}$ are isometric.
\end{lemma}

\begin{proof}
Let%
\[
k=\int_{0}^{\mu}\sqrt{\frac{(h+\mu)(h^{2}+\mu^{2})}{h^{3}}}dh\text{.}%
\]
In the notation of Lemma \ref{FronsdalImb}, $p\in\mathcal{H}_{2}$ if and only
if $\left\vert u_{2}(p)-u_{1}(p)\right\vert =0$, $u_{3}(p)=k$ and
$u_{4}(p)^{2}+u_{5}(p)^{2}+u_{6}(p)^{2}=\mu^{2}$. Let $v_{1}=$ $u_{1}%
|\mathcal{H}$, ..., $v_{6}=u_{6}|\mathcal{H}$. So
\[
dv_{1}=dv_{2},dv_{3}=0\text{,}%
\]
while
\[
dv_{4}\otimes dv_{4}+dv_{5}\otimes dv_{5}+dv_{6}\otimes dv_{6}=\mu^{2}%
\zeta_{S^{2}}\text{,}%
\]
where $\zeta_{S^{2}}$ the Euclidean metric of $S^{2}$. Thus the metric of
$\mathcal{H}$ induced from the Fronsdal manifold (cf. Definition
\ref{FronsdalManifold}) becomes%
\[
-dv_{1}\otimes dv_{1}+dv_{2}\otimes dv_{2}+dv_{3}\otimes dv_{3}+\mu^{2}%
\zeta_{S^{2}}=\mu^{2}\zeta_{S^{2}}%
\]
and the result follows from Lemma \ref{Horizon}.
\end{proof}

\bigskip

Hence, it follows from Lemmas \ref{Fronsdal} and \ref{Horizon 2} that

\begin{corollary}
The Fronsdal spacetime is the embedding in the Fronsdal manifold of the
Kruskal-Szekeres spacetime.
\end{corollary}

\section{Final Considerations \label{Final}}

We have explained in details in Section \ref{Solut} that the solutions of
Schwarzschild and Hilbert-Droste are indeed \textit{different} solutions of
Einstein equation because they possess a very different topology. And it is
important to stress that such a difference is not of a secondary importance
since it is possible, in principle, to ascertain which one corresponds to the
physical reality by verifying the existence of a black hole in the
gravitational field that those solutions describes (that is, in the presence
of an isolated point of mass).

Indeed, we have seen that the Schwarzschild solution has as a spacetime
manifold the $\mathbb{R}^{4}$ with the worldline of the particle generating
the field removed, that is, $\mathbb{R}\times\mathbb{R}^{+}\times S^{2}$.
Therefore, one cannot have a black hole in such spacetime and, even if a
process of maximal extension is formally possible, it would destroy the
original topology of the spacetime manifold, which was fixed \textit{a priori}
by Schwarzschild (compare this with the example discussed in the last
paragraph of this Section). On the other hand, the manifold of the
Hilbert-Droste solution have a disconnected topology that is given by
$\mathbb{R}\times\left(  \mathbb{R}^{+}-\{\mu\}\right)  \times S^{2}$, which
clearly cannot be a satisfactory spacetime model since it is impossible to
define a global time orientation in the whole manifold.

Therefore, a maximal extension of the Hilbert-Droste solution is required, as
was described in Section \ref{Exten-HD}. There, two approaches were studied.
One was the classic Kruskal-Szekeres spacetime, a maximal manifold that
contains two regions, one isometric to the black hole and the other to the
normal (or exterior) region of the Hilbert-Droste manifold, together with an
exotic submanifold homeomorphic to $[0,\infty\lbrack\times S^{2}$ connecting
the black hole to the exterior, generally known as wormhole. The other
approach which we covered was the Fronsdal imbedding of the Hilbert-Droste
\textquotedblleft spacetime\textquotedblright\ in a 6-dimensional vectorial
manifold by improving a procedure minted by Kasner almost four decades before.

On what follows, we shall finish our endeavour commenting some works in the literature.

We begin with an author who as able to appreciate most of what was told above,
the differential geometer N. Stavroulakis. He wrote a series of notes entitled
\emph{V\'{e}rit\'{e} scientifique et trous noirs }\cite{Nik}, published in the
\emph{Annales de la Fondation Louis de Broglie} in 1999. There, the geometer
presented a critical analysis of many practices usually employed by
relativists when studying solutions of Einstein equation with black holes.

For instance, recognizing that solutions with different manifolds actually
describes different physical situations, Stavroulakis wrote in his notes that

\begin{quotation}
\textquotedblleft Puisque la vari\'{e}t\'{e} n'est pas fix\'{e}e d'avance, la
pr\'{e}sentation d'une solution dans divers syst\`{e}mes de coordonn\'{e}es
locales dissimule souvent l'utilisation de vari\'{e}t\'{e}s diff\'{e}rentes.
Mais alors il s'agit d'un probl\`{e}me sans objet, car l'introduction de
vari\'{e}t\'{e}s distinctes donne lieu n\'{e}cessairement \`{a} des
probl\`{e}mes distincts\textquotedblright.
\end{quotation}

Stavroulakis was particularly discontented with the use of manifolds with
boundary and the use of what he called \textquotedblleft implicit
transformations\textquotedblright\ when solving Einstein equation.

Concerning the first, the geometer criticizes, for instance, the continuity
condition that Schwarzschild adopted to determine one of the constants which
appears in his solution (cf. Remark \ref{Constant}). Indeed, such a condition
requires the introduction of the manifold with boundary $\mathbb{R}%
\times\mathbb{[}0,\infty\lbrack\times S^{2}$, so that the metric tensor is
required to be continuous in the whole new spacetime manifold except at the
boundary\ $\{0\}\times S^{2}$.

However, since $\mathbb{R}\times\mathbb{[}0,\infty\lbrack\times S^{2}$ is not
homeomorphic to $\mathbb{R}^{4}$ nor to any submanifold of the latter, its
boundary $\{0\}\times S^{2}$ is absent of physical meaning, if one actually
\textit{believes} that the gravitational field of an isolated mass point must
be described by some submanifold of $\mathbb{R}^{4}$. This is not the case of
course of the Kruskal-Szekeres spacetime.

Another criticism of the use of manifolds with boundary stressed by
Stavroulakis is that, in the case of the Schwarzschild manifold, a Riemannian
\textquotedblleft metric\textquotedblright\ defined on $\mathbb{R}%
\times\mathbb{[}0,\infty\lbrack\times S^{2}$ which is continuous
(differentiable, respectively) and positive in $\mathbb{R}\times
\mathbb{]}0,\infty\lbrack\times S^{2}$, is however null at the boundary
$\{0\}\times S^{2}$ -- something which we may call a \textit{pseudo-metric} --
normally lead to a metric which is not continuous (differentiable,
respectively) at the origin of $\mathbb{R}^{4}$. A simple example can be given
in $\mathbb{[}0,\infty\lbrack\times S^{2}$ with a \textit{pseudo-metric}
defined by%
\[
2dr\otimes dr+r^{2}\zeta_{S^{2}}%
\]
where $r$ is the identity of $\mathbb{[}0,\infty\lbrack$ and $\zeta_{S^{2}}$
the Euclidean metric of $S^{2}$. Introducing natural coordinates
$(x^{i})_{i\in\lbrack1,3]}$ in $\mathbb{R}^{3}$, the metric induced in
$\mathbb{R}^{3}-\{0\}$ is%
\[%
{\displaystyle\sum\limits_{i\in\lbrack1,3]}}
dx^{i}\otimes dx^{i}+\frac{1}{\left\Vert x\right\Vert ^{2}}\left(
{\displaystyle\sum\limits_{i\in\lbrack1,3]}}
x^{i}dx^{i}\right)  \otimes\left(
{\displaystyle\sum\limits_{i\in\lbrack1,3]}}
x^{i}dx^{i}\right)
\]
where $\left\Vert x\right\Vert ^{2}=%
{\textstyle\sum}
(x^{i})^{2}$, which of course is undefined for $0\in\mathbb{R}^{3}$. The same
can be easily generalized for the manifold of interest, $\mathbb{R}%
\times\mathbb{[}0,\infty\lbrack\times S^{2}$. For example, the reader may
verify that, given mappings $f,g$ from $\mathbb{R}\times\mathbb{[}%
0,\infty\lbrack\times S^{2}$ into $\mathbb{R}$, both differentiable in
$\mathbb{R}\times\mathbb{]}0,\infty\lbrack\times S^{2}$, if we define in
$\mathbb{R}\times\mathbb{[}0,\infty\lbrack\times S^{2}$ the Bondi's
pseudo-metric,%
\[
-\exp(2f)dt\otimes dt-\exp(f+g)\left(  dt\otimes dr+dr\otimes dt\right)
+r^{2}\zeta_{S^{2}}%
\]
where $(t,r)$ is the natural coordinate system of $\mathbb{R}\times
\mathbb{[}0,\infty\lbrack$ and $\zeta_{S^{2}}$ as before, then, after
expressing the latter metric in the natural coordinates of $\mathbb{R}^{4}$,
the same cannot be defined for $0\in\mathbb{R}^{4}$.

So Stavroulakis argues

\begin{quotation}
\textquotedblleft Dans de telles situations la diff\'{e}rentiabilit\'{e} de la
forme consid\'{e}r\'{e}e sur $\mathbb{[}0,\infty\lbrack\times S^{2}$ est
illusoire, car elle dissimule les singularit\'{e}s de la forme d'origine sur
$\mathbb{R}^{3}$. La g\'{e}ometrie diff\'{e}rentielle classique ne prend pas
en consid\'{e}ration les situations de ce genre qui n\'{e}cessitent une
\'{e}tude \`{a} part afin d'\'{e}lucider la nature des singularit\'{e}s. En ce
qui concerne la relativit\'{e} g\'{e}n\'{e}rale, on ne saurait introduire des
m\'{e}triques comportant des singularit\'{e}s
g\'{e}n\'{e}riques\textquotedblright.
\end{quotation}

However, at this point, one may take a position different from that of
Stavroulakis and argue that the Schwarzschild solution is a reasonable
physical model, even with the Schwarzschild's continuity condition making
reference to the boundary $\{0\}\times S^{2}$. That is because the
\textquotedblleft singularity\textquotedblright\ that such condition
introduces in the solution (in $\mathbb{R\times R}^{3}$) is along the
worldline of the particle generating the field, what may be seen as physically acceptable.

The situation here is similar to classical mechanics, when one removes a
finite set of points in $\mathbb{R}^{3}$ in order to deal with problems
involving particles interacting through a Newtonian potential. (Indeed, as it
is well known, the only situation that such a description encounters difficult
in classical mechanics is in the presence of collisions). But we shall not
enter on this discussion here and the interested reader must consult
\cite{Nik}.

\begin{remark}
As Stavroulakis himself believed that the solution which describes the
gravitational field of a spherically symmetric body must have $\mathbb{R\times
R}^{3}$ as spacetime manifold, he proposed his own solution in \cite{Nik1},
where he argues that the existence of point of mass is a hypothesis
incompatible with General Relativity on the grounds that such objects
introduces singularities in the solution.
\end{remark}

On the issue of the \textquotedblleft implicit
transformations\textquotedblright, Stavroulakis said that

\begin{quotation}
\textquotedblleft Une transformation implicite est cens\'{e}e \^{e}tre
d\'{e}finie par un syst\`{e}me d'\'{e}quations (\'{e}quations ordinaires pour
la d\'{e}finition de fonction implicites, \'{e}quations diff\'{e}rentielles,
\'{e}quations aux d\'{e}riv\'{e}es partielles) contenant les composantes
inconnues du tenseur m\'{e}trique (...). Or la solution effective d'un tel
syst\`{e}me ne pourrait \^{e}tre envisag\'{e}e que si les composantes en
question \'{e}taient connues. Par cons\'{e}quent les transformations
implicites sont des transformations hypoth\'{e}tiques dont l'existence
m\^{e}me sur $U$ (ou \'{e}ventuellement sur un ouvert contenu dans $U$) n'est
pas assur\'{e}e\textquotedblright.
\end{quotation}

The \textquotedblleft implicit transformations\textquotedblright, defined by
Stavroulakis in the above excerpt, were employed in the construction of the
Hilbert-Droste solution, as the reader can recall from Subsection
\ref{HDSolution}. Indeed, the special coordinates $(t,h)$ of the Schwarzschild
plane $\Pi_{M}$ (check Definitions \ref{HDModel} and \ref{Plane}) were chosen
to be such that the metric component $\alpha$ in the warped product $\Pi
_{M}\times_{\alpha}S^{2}$ obeys the condition $\alpha\circ h=\operatorname{id}%
_{\mathbb{R}}$. That is, our coordinate system was chosen in such a way that
the metric could be writing as
\begin{equation}
-(f\circ h)dt\otimes dt+(g\circ h)dh\otimes dh+h^{2}\zeta_{S^{2}}
\label{HDExpression}%
\end{equation}
because $h\mapsto\alpha(h)=h$. So, according to Stavroulakis, this implicit
transformation is the origin of the famous singularity at $h=\mu$. As he wrote
in \cite{Nik2} (with our notation and enumeration)

\begin{quotation}
\textquotedblleft As is expected, the solution of the Einstein equations
related to \ref{HDExpression} is static:%
\[
-\left(  1-\frac{\mu}{h}\right)  dt\otimes dt+\frac{1}{1-\mu/h}dh\otimes
dh+h^{2}\zeta_{S^{2}}%
\]
In fact it is the Droste solution, or, more precisely, the Droste-Hilbert
solution, wrongly called Schwarzschild's solution in the literature. We have
already seen that the implicit diffeomorphism considered (...) is in general
actually inexistent. Now the discontinuity of the Droste solution at $h=\mu$
proves that the implicit diffeomorphism in question is also inconsistent with
the differentiable solutions of the Einstein equations\textquotedblright.
\end{quotation}

And, as repeated many times in our work, the \textquotedblleft singularity
$h=\mu$\textquotedblright\ causes the manifold of the Hilbert-Droste solution
to be disconnected, so that we can conclude that the use of an implicit
transformation is the origin of the necessity of the maximal extension of the
Hilbert-Droste solution, the Kruskal-Szekeres spacetime.

As a matter of fact, the maximal extension was also a target of Stavroulakis
criticism, as we can see from \cite{Nik},

\begin{quotation}
\textquotedblleft L'introduction des vari\'{e}t\'{e}s \`{a} bord \`{a}
entra\^{\i}n\'{e} l'id\'{e}e bizarre d'extension maximale. Celle-ci est vide
de sens par rapport \`{a} la vari\'{e}t\'{e} $\mathbb{R\times R}^{3}$. En ce
qui concerne l'extension de $\mathbb{R}\times\mathbb{[}0,\infty\lbrack\times
S^{2}$ au sens de Kruskal, elle n\'{e}cessite des identifications au moyen
d'applications discontinues qui ne sont pas math\'{e}matiquement
autoris\'{e}es\textquotedblright.
\end{quotation}

or, from the same source,

\begin{quotation}
\textquotedblleft La m\'{e}thode de Kruskal elle-m\^{e}me comporte des
incoh\'{e}rences qui transgressent les principes \'{e}l\'{e}mentaires des
raisonnements math\'{e}matiques\textquotedblright.
\end{quotation}

We, of course, cannot agree with the idea that the maximal extension of the
Hilbert-Droste solution is mathematically inconsistent, since we have
dedicated Section \ref{Exten-HD} to two distinct approaches to the extension
of the Hilbert-Droste manifold, that of Kruskal and Szekeres and the imbedding
of Kasner and Fronsdal. So now we must understand on what grounds Stavroulakis
based his latter remarks.

In fact, the geometer was referring to the idea of an \textquotedblleft
apparent singularity\textquotedblright, disseminated in the Relativity
community and whose germ can be found in the original papers of Synge, Kruskal
and Szekeres, as, e.g., one can see from Synge's letter to the editor
\cite{Syn}, partly quoted above, (in our notation)

\begin{quotation}
\textquotedblleft The existence of the singularity at $h=\mu$ is strange and
demands investigation. Investigation shows that there is no singularity at
$h=\mu$; the apparent singularity in the Schwarzschild line element is due to
the coordinates employed and may be removed by a transformation, so that there
remains no singularity except at $h=0$. This was pointed out by Lema\^{\i}tre
in 1933\textquotedblright.
\end{quotation}

A singularity is classified as \textquotedblleft apparent\textquotedblright%
\ if it can be attributed to a bad choice of coordinate system, and can be
introduced or removed from a coordinate expression of the metric through a
coordinate change. However, as Stavroulakis argued, this kind of
transformation cannot be a diffeormophism, being therefore not a permissible
coordinate transformation.

For instance, from a formal point of view, the \textquotedblleft singularity
$h=\mu$\textquotedblright\ of the Hilbert-Droste solution cannot be removed by
means of a coordinate transformation like Equations \ref{CoordChange1} and
\ref{CoordChange2} used in the proof of Lemma \ref{KruskalLemma} or the one
used originally by Kruskal, given by Equations \ref{CoordChange3} and
\ref{CoordChange4}, since that both transformations are degenerated exactly
at, in the notation of this section, $h=\mu$.

What we can conclude from Stavroulakis remarks, however, is that, from a
formal perspective, the maximal extension of a solution of Einstein equation
is not executed by means of a coordinate change, but by \emph{postulating} a
manifold, like the Kruskal-Szekeres manifold, and showing that there exists
submanifolds belonging to the former which are isometric to the submanifolds
of the solution in question, e.g., the black hole and the normal region of the
Hilbert-Droste \textquotedblleft spacetime\textquotedblright.

On the other hand, not everyone understood so well the issues presented in our
work. For instance, referring to some modern concepts in General Relativity,
like \textquotedblleft event horizon\textquotedblright, \textquotedblleft
essential singularity\textquotedblright\ and, we empathize,
\emph{\textquotedblleft maximal extension\textquotedblright}, Ll. Bel wrote in
\cite{Bel} that

\begin{quotation}
\textquotedblleft All this is nice geometry in the making but the point is
that none of this is as yet necessary to understand that Schwarzschild's
original work is a better piece of physics than the extravaganzas to which one
is led with some of the extensions of Schwarzschild's
solution\textquotedblright.
\end{quotation}

If Ll. Bel was referring to the original Schwarzschild solution, as he claim
to be, then he would be completely right since the spacetime manifold of the
latter solution is entirely satisfactory on its own, being connected and
dispensing any process of maximal extension. Unfortunately, however, as Bel
confuses the original \textquotedblleft Schwarzschild\textquotedblright%
\ solution with that of Hilbert and Droste, we believe that he must be wrong
because, as stressed many times in our work, the Hilbert-Droste manifold
cannot define a proper spacetime -- in the sense of Definition \ref{Spacetime}
(see also Remark \ref{TimeMotivation} for a motivation of the latter
Definition) -- and therefore, a maximal extension is not any \textquotedblleft
extravaganza\textquotedblright, but a necessary procedure if one is willing to
accept the Hilbert-Droste \textquotedblleft spacetime\textquotedblright\ as a
description of Nature.

Indeed, comparing the coordinate expressions for the metric in the
Schwarzschild solution to the one in the Hilbert-Droste solution, Bel wrote in
that same article that

\begin{quotation}
\textquotedblleft This new form [\emph{the Hilbert-Droste metric}] is simpler
to obtain than (1) [\emph{the Schwarzschild metric}] and also simpler to write
down and is the form which is used overwhelmingly in textbooks. Notice that it
can be derived directly from Schwarzschild's form following two different, but
equivalent, paths:

(i) To use the definition (3)\footnote{In our notation, $R\equiv
\alpha(h)=\left(  3h+\mu^{3}\right)  ^{1/3}$. Recall Proposition
\ref{SchwarzschildSolution}.} of the auxiliary function $R$ as a coordinate
transformation and get rid of the spurious parameter, or (ii) choose for
simplicity $\rho$ = 0\footnote{In the notation of Proposition
\ref{Schwarzschild}, $k=0$.}\textquotedblright.
\end{quotation}

Then he explains why Schwarzschild could not follow his \textquotedblleft path
(i)\textquotedblright,

\begin{quotation}
\textquotedblleft Schwarzschild could not follow the first path because he
thought he was dealing with a theory which did not allow arbitrary coordinate
transformations (but in fact he had already done it when he abandoned his
initial coordinates for those used in (1))\textquotedblright.
\end{quotation}

And Bel is not alone in his opinion. As P. Fromholz \textit{et al}. wrote in
\cite{Fro},

\begin{quotation}
\textquotedblleft Schwarzschild noticed that by defining a new variable%
\[
r_{s}\equiv(3x+b)^{1/3}=(\rho^{3}+b)^{1/3}%
\]
he could put the metric (6) \emph{[referring to Schwarzschild metric]} into a
simpler form, which is precisely Eq. (4) \emph{[the Hilbert-Droste
metric]\textquotedblright.}
\end{quotation}

But the authors of the latter article went a little further showing a complete
ignorance about the topology which Schwarzschild himself fixed in his
spacetime manifold, writing that

\begin{quotation}
\textquotedblleft But Schwarzschild went on to address the integration
constant $b$. He demanded that the metric be regular everywhere except at the
location of the mass-point, which he assigned to be at $\rho=0$, where the
metric should be singular. This fixed $b=(2M)^{3}$. This choice resulted in
considerable confusion about the nature of the \textquotedblleft Schwarzschild
singularity\textquotedblright, which was not cleared up fully until the 1960s.
Because we now are attuned to the complete arbitrariness of coordinates, we
understand that $\rho=0$, or $r_{s}=2M$ is not the origin, but is the location
of the event horizon, while $\rho=-2M$, or $r_{s}=0$ is the location of the
true physical singularity inside the black hole\emph{\textquotedblright. }
\end{quotation}

Observe that Fromholz \textit{et al}. talks about an \textquotedblleft event
horizon\textquotedblright\ and a \textquotedblleft true physical singularity
inside the black hole\textquotedblright\ in a solution whose manifold is
giving by $\mathbb{R}\times\mathbb{R}^{+}\times S^{2}$, which is already a
nonsense. But this is not the worse part yet. Now, when these authors wrote
\textquotedblleft$\rho=-2M$\textquotedblright, they completely disrespected
the topology (and even the domain of definition of the chart) chosen by
Schwarzschild in his original paper, since that $x$ (in Fromholz \textit{et
al}. notation and $x_{1}$ according to Schwarzschild paper), representing the
radial coordinate of $\mathbb{R}\times\mathbb{R}^{+}\times S^{2}$ -- in fact,
$x=r^{3}/3$ -- must be a positive real number, so that $\rho=\sqrt[3]{3x}>0$.

\begin{remark}
It is worth mentioning that Schwarzschild could not choose $k=0$ \emph{(}our
notation; see footnote of latter page\emph{),} as suggested by Bel, because it
would give a solution incompatible with the manifold fixed by Schwarzschild.
That is because, with the condition $k=0$, the coordinate expression for the
metric would have a singularity in $R=\mu$ \emph{(}our notation\emph{)},
something that can be interpreted in two ways. First, if such a singularity is
seen as a property of the metric tensor, it would not be satisfactory because
the set of points for which $R=\mu$ is contained in the Schwarzschild
manifold. Second, if the singularity is interpreted as \textquotedblleft
apparent\textquotedblright\ because of a bad choice of a coordinate system
\emph{(}the usual perspective today\emph{)}, then the spacetime manifold of
the solution could not be $\mathbb{R}\times\mathbb{R}^{+}\times S^{2}$
because, as it was proved in Section \emph{\ref{Exten-HD}}, the maximal
extension of such a solution \emph{(}the Kruskal-Szekeres spacetime\emph{)}
have an exotic topology totally different from $\mathbb{R}\times\mathbb{R}%
^{+}\times S^{2}$.
\end{remark}

J. Senovilla left a similar opinion in his rectification note \cite{Sen}
concerning the equivalence of Schwarzschild solution and that of Hilbert-Droste:

\begin{quotation}
\textquotedblleft I would like to remark here that Karl Schwarzschild did
write the form (1) [\emph{referring to the Hilbert-Droste coordinate
expression}] of the metric: see formula (14) in the GRG Golden Oldie
translation [16](b). The myth that he did not do it must be dispelled. To
argue that the $R$ in that formula was in fact a function of the radial
coordinate that he used -- due to the famous story of the unit-determinant
gauge choice favoured by Einstein at early stages -- is completely irrelevant
today, given the general covariance of the theory and, especially, the fact
that Schwarzschild wrote \textquotedblleft$dR^{2}$\textquotedblright\ and
expressed the whole line-element in terms of $R$
exclusively.\textquotedblright
\end{quotation}

However, the fact that General Relativity is \textquotedblleft generally
covariant\textquotedblright\ or, in more precise terms, diffeomorphic
invariant, does not means that one can arbitrarily change the topology of the
spacetime manifold, but only the coordinate system. That is, even if there
exists a coordinate transformation (which was described here in Remark
\ref{CoordExpression}) that transforms the coordinate expression for the
Schwarzschild metric to one with the same form as the Hilbert-Droste
expression, one cannot jump to the conclusion that these solutions are indeed
the same.

And, differently from some of the latter authors\textit{, }if we pay attention
to the domains of the coordinate transformation of Remark
\ref{CoordExpression}, it must be clear that the diffeomorphism $h\mapsto
R=\alpha(h)=\left(  3h+\mu^{3}\right)  ^{1/3}$, from $]0,\infty\lbrack$ onto
$]\mu,\infty\lbrack$, transforms the coordinate expression of the
Schwarzschild metric (notation as in Section \ref{Solut})%
\[
-\left[  1-\frac{\mu}{\alpha(h)}\right]  dt\otimes dt+\frac{1}{[\alpha
(h)]^{4}}\frac{1}{1-\mu/\alpha(h)}dh\otimes dh+[\alpha(h)]^{2}\zeta_{S^{2}}%
\]
to one reassembling the Hilbert-Droste metric,
\[
-\left(  1-\frac{\mu}{R}\right)  dt\otimes dt+\frac{1}{1-\mu/R}dR\otimes
dR+R^{2}\zeta_{S^{2}}%
\]
\emph{which holds, however, only for }$R>\mu$, since that $\alpha
(]0,\infty\lbrack)=$ $]\mu,\infty\lbrack$. That is, the \textquotedblleft
interior\textquotedblright\ $R<\mu$ is meaningless in the Schwarzschild
solution, \textit{as it describes a manifold which is completely disconnected
from the former}.

\begin{remark}
The reason for which the diffeomorphism $h\mapsto R=\alpha(h)$ was defined on
$]0,\infty\lbrack$ is that $h$ belongs to the natural coordinate system of
what we called above the Schwarzschild plane $P=\mathbb{R}\times\mathbb{R}%
^{+}$. Indeed, $h$ can be interpreted as the radial coordinate of the
Schwarzschild spacetime manifold $\mathbb{R}\times\mathbb{R}^{+}\times S^{2}$,
which is possible since Schwarzschild fixed his manifold \textbf{a priori}.{}
\end{remark}

We finish by discussing a trivial example that illustrates very well what was
told above and can be found in the same paper where Szekeres presented his
maximal extension of the Hilbert-Droste solution (cf. ref. \cite{Sze}), almost
five decades ago. Let $\mathbb{R}^{+}\times S^{2}$ be giving with the
structure of an Euclidean manifold, that is, with the metric giving by%
\[
g=dh\otimes dh+h^{2}\zeta_{S^{2}}%
\]
where $h$ is the identity (or the natural coordinate) of $\mathbb{R}^{+}$ and
$\zeta_{S^{2}}$ is the Euclidean metric of $S^{2}$. Let $\mu>0$ be a real
number. So, with the diffeomorphism $h\mapsto R=h+\mu$ from $\mathbb{R}^{+}$
onto $]\mu,\infty\lbrack$, the coordinate expression for $g$ can be written as%
\begin{equation}
g=dR\otimes dR+\left(  R-\mu\right)  ^{2}\zeta_{S^{2}}\label{I}%
\end{equation}
which holds, \textit{as in the Schwarzschild case}, only for $R>\mu$. That is,
the Euclidean space which we started with is now identified with $R>\mu$, in
such a way that speak about the \textquotedblleft region\textquotedblright%
\ $R<\mu$ here -- without changing our original manifold -- is simply a
\textit{nonsense}, because it is not only outside the domain of definition of
our diffeomorphism, as it describes a submanifold\ which is disconnected from
$\mathbb{R}^{+}\times S^{2}$, having nothing to do with our ordinary Euclidean
topology. As Szekeres remarked in ref. \cite{Sze} (with our notation),

\begin{quotation}
\textquotedblleft Here we have an apparent singularity on the sphere $R=\mu$,
due to a spreading out of the origin over a sphere of radius $\mu$. Since the
exterior region $R>\mu$ represents the whole of Euclidean space (except the
origin), the interior $R<\mu$ is entirely disconnected from it and represents
a distinct manifold\textquotedblright.
\end{quotation}

If one insists, however, we can still use the expression for the metric given
by Eq. (\ref{I}) for all $R\in$ $]0,\infty\lbrack-\{\mu\}$, but only
\textit{if we pay the price of changing our original manifold}. That is, we
may define a new manifold $M=\left(  \mathbb{R}^{+}-\{\mu\}\right)  \times
S^{2}$, which is not equivalent to the Euclidean $\mathbb{R}^{+}\times S^{2}$,
and give to it a Riemannian structure whose metric is defined by%
\[
g^{\prime}=dR\otimes dR+\left(  R-\mu\right)  ^{2}\zeta_{S^{2}}%
\]
for all positive real number $R\neq\mu$. The fact that the Riemannian
structures $(\mathbb{R}^{+}\times S^{2},g)$ and $(M,g^{\prime})$ are different
from each other (even if the former contains a submanifold isometric to the
latter) are very far from being controversial -- since there is no black hole
or any other polemical issue in the present game -- but is, on the other hand,
of the same nature as the difference between the solutions of Schwarzschild
and Hilbert-Droste,\ which are polemical subjects in the current literature.

\section{Acknowledgement}

The author is grateful to Professor Waldyr A. Rodrigues Jr. for his thoughtful
advisement and the suggestion to study the issues covered in the paper.

\appendix{}

\section{Topological Extension of Manifolds\label{Exten}}

In this Appendix, we analyze the extension of manifolds from a careful
topological point of view. Specifically, we give a rigorous procedure
(summarized in the following list of Definitions and Lemmas) that justify the
process of gluing topological spaces, manifolds and pseudo-Riemannian structures.

We hope that the following developments might be useful for relativists
working in the construction of spacetimes containing black holes, wormholes,
bridges or any object with exotic topology.

Our approach is based in ref. \cite{On2}.

\begin{definition}
A gluing structure is an ordered list $(M,N,U,V,\xi)$, where $M$ and $N$ are
topological spaces, $U\subset M$ and $V\subset N$ are subspaces and $\xi$ is a
homeomorphism between $U$ and $V$.
\end{definition}

Recall that, given a family of sets $(A_{n})_{n\in F}$, the disjoint union of
this family is defined to be%
\[
\widetilde{\mathbf{\cup}}_{n\in F}A_{n}=\cup_{i\in F}\{(x,i):x\in A_{i}\}
\]

\begin{definition}
\label{gluedspace}Let $G=(M,N,U,V,\xi)$ be a gluing structure. Define $\eqsim$
to be the equivalence relation on the disjoint union $M\widetilde{\cup}N$ such
that $p\eqsim q$ if and only if $p=q$, $p=\xi(q)$ or $q=\xi(p)$. So the
quotient space $M\widetilde{\cup}N/\eqsim$ will be called the glued space
$Q_{G}$ of $G$ and $\eqsim$ the equivalence of the gluing structure $G$.
\end{definition}

In what follows, given a gluing structure $G=(M,N,U,V,\xi)$ and its respective
glued space $Q_{G}$, the natural injections $i$ and $j$ from $G$ into $Q_{G}$
are the mappings from $M$ and $N$, respectively, into $Q_{G}$ such that

\begin{quotation}
$i(p)=p$ if $p\in M-U$ and $i(p)=\{p,\xi(p)\}$ if $p\in U$, $j(q)=q$ if $q\in
N-V$ and $j(q)=\{q,\xi(q)\}$ if $q\in V$.
\end{quotation}

A subset $S\subset Q_{G}$ will be considered open if and only if $i^{-1}(S)$
and $j^{-1}(S)$ are open in $M$ and in $N$, respectively.

\begin{lemma}
\label{Homeo}Let $G$ be a gluing structure, $Q_{G}$ its glued space and $i$
and $j$ the natural injections from $G$ into $Q_{G}$. Then $i$ and $j$ are
homeomorphisms between $M$ and $i(M)$ and between $N$ and $j(N)$, respectively.
\end{lemma}

\begin{proof}
By the last remark, $i$ and $j$ are continuous. Let $X\subset M$ be an open
subset. So $i(X)$ is open in $Q_{G}$ if and only if $i^{-1}(i(X))$ and
$j^{-1}(i(X))$ are open in $M$ and in $N$ respectively. The first is open
since that $i^{-1}(i(X))=X$. But for the second:%
\[
j^{-1}(i(X))=j^{-1}(i(X)\cap j(N))=j^{-1}(i(X\cap U))=\xi(X\cap U)
\]
Hence $i(X)$ is open. The result follows for $i$ since that it is injective,
and the proof is the same for $j$.
\end{proof}

\begin{remark}
Because of the last Lemma, one may ignore the natural injections and think
about $i(M)$ and $j(N)$ as being actually equal to $M$ and $N$, respectively.
Then $M\cap N$, $U$ and $V$ are all identified.
\end{remark}

\begin{lemma}
Assuming the hypothesis of Lemma \emph{\ref{Homeo}}, let $P$ be a topological
space and let $\phi_{M}$ and $\phi_{N}$ be continuos mappings from $M$ and
$N$, respectively, into $P$. Suppose that $\phi_{M}|U=\phi_{N}\circ\xi$.
Hence, there is a unique continuos mapping $\phi$ from $Q_{G}$ into $P$ such
that $\phi\circ i=\phi_{M}$ and $\phi\circ j=\phi_{N}$.
\end{lemma}

\begin{proof}
Define $\phi$ to be such that $\phi(p)=\phi_{M}(i^{-1}(p))$ if $p\in i(M)$ and
$\phi(q)=\phi_{N}(j^{-1}(q))$ if $q\in j(N)$. This is well-defined since that
when $p\in i(M)\cap j(N)$, $p=\{x,\xi(x)\}$ for $x=i^{-1}(x)$. Hence%
\[
\phi(p)=\phi_{M}(x)=\phi_{N}(\xi(x))=\phi(p)
\]
Finally, $\phi$ is continuos by Lemma \ref{Homeo}.
\end{proof}

\begin{lemma}
Let $G=(M,N,U,V,\xi)$ and $G^{\prime}=(M^{\prime},N^{\prime},U^{\prime
},V^{\prime},\xi^{\prime})$ be gluing structures, $Q_{G}$ and $Q_{G}^{\prime
}=Q_{G^{\prime}}$ their respective glued spaces and $i,j$ and $i^{\prime
},j^{\prime}$ their respective natural projections. Let $\phi_{M}$ and
$\phi_{N}$ be continuos mappings from $M$ and $N$, respectively, into
$M^{\prime}$ and $N^{\prime}$, respectively. Assume that $\xi^{\prime}%
\circ\phi_{M}|U=\phi_{N}\circ\xi$. Thus, there is a unique continuous mapping
$\phi$ from $Q_{G}$ into $Q_{G}^{\prime}$ such that $\phi\circ i=i^{\prime
}\circ\phi_{M}$ and $\phi\circ j=j^{\prime}\circ\phi_{N}$.
\end{lemma}

\begin{proof}
Define $\phi$ to be such that $\phi(p)=i^{\prime}(\phi_{M}(i^{-1}(p)))$ if
$p\in i(M)$ and $\phi(q)=j^{\prime}(\phi_{N}(j^{-1}(q)))$ if $q\in j(N)$. This
is well-defined since that, if $p\in i(M)\cap j(N)$, $p=\{x,\xi(x)\}$ for
$x=i^{-1}(p)$. So%
\[
\phi(p)=i^{\prime}(\phi_{M}(x))=(j^{\prime}\circ\xi^{\prime})(\phi
_{M}(x))=j^{\prime}(\phi_{N}(\xi(x)))=\phi(p)
\]
Finally, $\phi$ is continuos by Lemma \ref{Homeo}.
\end{proof}

\bigskip

\bigskip The last two Lemmas are normally called the \textit{Mapping Lemmas}.

\begin{exercise}
\label{HausdorffExercises}(a) Let $G=(\mathbb{R},\mathbb{R},\mathbb{R}%
^{+},\mathbb{R}^{+},\operatorname{id}_{\mathbb{R}})$ be a gluing structure.
\emph{(}For any set $A$, $\operatorname{id}_{A}$ means the identity in
$A$\emph{)}. Is $Q_{G}$, the glued space, Hausdorff? (b) Let $H^{+}$ be the
north hemisphere of $S^{2}$ without the equator, and let $N\in S^{2}$ be its
north pole. Let $G=(S^{2},S^{2},H^{+}-\{N\},H^{+}-\{N\},\operatorname{id}%
_{S^{2}})$ be a gluing structure. Is $Q_{G}$ Hausdorff?
\end{exercise}

The above exercise is then the motivation for the following definition:

\begin{definition}
\label{Haus}A gluing structure $(M,N,U,V,\xi)$ will be called Hausdorff if $M$
and $N$ are Hausdorff and if there is no convergent sequence $(p_{n}%
)_{n\in\mathbb{N}}$ of points in $M$ such that $\lim p_{n}\in M-U$ and
$\lim\xi(p_{n})\in N-V$.
\end{definition}

\begin{lemma}
Let $(M,N,U,V,\xi)$ be a Hausdorff gluing structure. So the glued space
$Q_{G}$ is Hausdorff.
\end{lemma}

\begin{proof}
Let $x,y\in Q_{G}$ be distinct points. The result is obvious if both $x$ and
$y$ belongs to $i(M)$ (or to $j(N)$). So, suppose that $x\in i(M)-j(V)$ and
$y\in j(N)-i(U)$. Let $(N_{n})_{n\in\mathbb{N}}$ and $(N_{n}^{\prime}%
)_{n\in\mathbb{N}}$ be a basis for the neighborhoods of $x$ and $y$,
respectively. Assume that $N_{n}\cap N_{n}^{\prime}$ is not empty for all $n$.
So by the axiom of choice, there is a sequence $(x_{n})_{n\in\mathbb{N}}$ such
that $x_{n}\in N_{n}\cap N_{n}^{\prime}$ for all $n$. Let $i$ and $j$ be the
natural injections of $G$ into $Q_{G}$. So $(i^{-1}(x_{n}))_{n\in\mathbb{N}}$
and $(j^{-1}(x_{n}))_{n\in\mathbb{N}}$ do not respect Definition \ref{Haus}.
Hence, by contradiction, there is some $n\in\mathbb{N}$ such that $N_{n}\cap
N_{n}^{\prime}=\varnothing$, and the proof is over.
\end{proof}

\begin{exercise}
Let $U=V=\{(x,y)\in\mathbb{R}^{2}:x,y<0\}$ and $G=(\mathbb{R}^{2}%
,\mathbb{R}^{2},U,V,\xi)$. Is the glued space $Q_{G}$ Hausdorff if \emph{(a)}
$\xi=\operatorname{id}_{\mathbb{R}^{2}}$ and \emph{(b)} $\xi(x,y)=(x,y/x)$?
\end{exercise}

We can now extrapolate our results for manifolds:

\begin{definition}
A gluing structure $G=$ $(M,N,U,V,\xi)$ will be called a \textquotedblleft
manifold gluing\textquotedblright\ when $G$ is Hausdorff, $M$ and $N$ are
manifolds with the same dimension, $U$ and $V$ are submanifolds and $\xi$ is a diffeomorphism.
\end{definition}

\begin{remark}
\label{Smooth}It must be clear from the above definition that, in the
\textquotedblleft manifold gluing\textquotedblright\ case, the Mapping Lemmas
holds for smooth mappings rather than just for continuos ones.
\end{remark}

Remember that a chart in a manifold $M$ is an ordered pair $(X,\psi)$ such
that $X\subset M$ is an open subset and $\psi$ is a homeomorphism between $X$
and $\mathbb{R}^{\dim M}$. Recall also that an atlas in $M$ is a set $A$ of
charts such that $M\subset\cup_{(U,\psi)\in A}U$ (we say that $A$
\textit{covers} $M$) and, given two charts $(X,\psi),(Y,\omega)\in A$ such
that $X\cap Y\neq\varnothing$, both $\psi\circ\omega^{-1}$ and $\omega
\circ\psi^{-1}$ are smooth (we say that $A$ \textit{overlaps smoothly}).

\begin{lemma}
Let $G=$ $(M,N,U,V,\xi)$ be a \textquotedblleft manifold
gluing\textquotedblright. The glued space $Q_{G}$ is itself a manifold.
\end{lemma}

\begin{proof}
By hypothesis, $Q_{G}$ is Hausdorff. Now, let $A_{M}$ and $A_{N}$ be atlases
for $M$ and $N$ respectively, and define%
\[
A=\{(i(X),\psi\circ i^{-1}),(j(Y),\omega\circ j^{-1}):(X,\psi)\in
A_{M},(Y,\omega)\in A_{N}\}.
\]
Of course that $A$ \textit{covers} $Q_{G}$. To prove that they \textit{overlap
smoothly}, let $(X,\psi)\in A_{M}$ and $(Y,\omega)\in A_{N}$ be charts such
that $i(X)\cap j(Y)\neq\varnothing$. So%
\[
(\psi\circ i^{-1})\circ(\omega\circ j^{-1})^{-1}=\psi\circ(i^{-1}\circ
j)\circ\omega^{-1}%
\]
is smooth and the proof is over.
\end{proof}

\bigskip

\bigskip Finally, we finish with the pseudo-Riemannian case:

\begin{definition}
\label{pseudo-RiemGlu}A Hausdorff gluing structure $G=$ $(M,N,U,V,\xi)$ will
be called a pseudo-Riemannian gluing if $M$ and $N$ have pseudo-Riemaniann
structures \emph{(}see Definition \emph{\ref{pRm})} and if $\xi$ is an isometry.
\end{definition}

\begin{proposition}
Let $G=$ $(M,N,U,V,\xi)$ be a pseudo-Riemannian gluing and $g_{M}$ and $g_{N}$
the metric tensors of $M$ and $N$, respectively. So, there is an unique metric
tensor $g_{G}$ such that $(Q_{G},g_{G})$ is a pseudo-Riemannian manifold.
\end{proposition}

\begin{proof}
Let $i$ and $j$ be the natural projections of $G$ into $Q_{G}$ and let
$V,W\in\sec TQ_{G}$. So the mappings $x\rightarrow\phi_{M}^{(V,W)}%
(x)=g_{M}(i_{\ast}^{-1}(V_{x}),i_{\ast}^{-1}(W_{x}))$, from $M$ into
$\mathbb{R}$, and $y\rightarrow\phi_{N}^{(V,W)}(y)=g_{N}(j_{\ast}^{-1}%
(V_{y}),j_{\ast}^{-1}(W_{y}))$, from $N$ into $\mathbb{R}$, are smooth. By the
\textit{Mapping Lemmas} and Remark \ref{Smooth}, there is a unique smooth
mapping $p\rightarrow\phi^{(V,W)}(p)$ from $Q_{G}$ into $\mathbb{R}$ such that
$\phi^{(V,W)}\circ i=\phi_{M}^{(V,W)}|U$ and $\phi^{(V,W)}\circ j=\phi
_{N}^{(V,W)}|V$. Hence, just define $g_{G}\in\sec T^{2}Q_{G}$ to be such that
$g_{G}(V_{p},W_{p})=\phi^{(V,W)}(p)$ and it is left to the reader to show why
$g_{G}$ is smooth.
\end{proof}

\bigskip

Thus the title of this Appendix is justified by the fact that $Q_{G}$ can be
called, suggestively, the extension of the manifolds $M$ and $N$.

\section{Einstein-Rosen Bridge\label{ERB}}

In this Appendix, we shall comment briefly on the mathematical realization of
the Einstein-Rosen bridge (ER bridge for short). Also, the well-known fact
that one cannot \textquotedblleft travel\textquotedblright\ through that
bridge, and even from one \textquotedblleft universe\textquotedblright\ of the
Kruskal-Szekeres spacetime to another (i.e., from $\mathcal{R}_{I}^{+}$ to
$\mathcal{R}_{I}^{-}$ or from $\mathcal{R}_{II}^{+}$, to $\mathcal{R}_{II}%
^{-}$), will be proven.

A. Einstein and N. Rosen published in 1935 an article entitled
\textit{\textquotedblleft The Particle Problem in the General Theory of
Relativity\textquotedblright} (cf. ref. \cite{ER}). There, the authors
proposed to eliminate the \textquotedblleft$r=\mu$\textquotedblright%
\ singularity of the Hilbert-Droste solution by introducing the idea that
elementary particles, in particular the electron, are an exotic topological
deformations of the spacetime manifold.

The traditional and heuristic construction of the ER bridge, which can be
found in any standard text on the subject of wormholes (cf. ref. \cite{Vis}),
proceeds as follows. One starts with the HD manifold $M=\mathbb{R}%
\times(\mathbb{R}^{+}-\{\mu\})\times S^{2}$ for some positive real $\mu$ and
define on it the HD metric%
\[
-\left(  1-\frac{\mu}{r}\right)  dt\otimes dt+\frac{1}{1-\mu/r}dr\otimes
dr+r^{2}\zeta_{S^{2}}\text{,}%
\]
where $(t,r)$ are natural coordinates of $\mathbb{R}^{2}$ restricted to $M$
and, as usual, $\zeta_{S^{2}}$ is the Euclidean metric of $S^{2}$. Now, the
\textit{\textquotedblleft}coordinate transformation\textit{\textquotedblright%
}\ $r\longmapsto u^{2}=r-\mu$ is introduced, and it is \emph{claimed} that the
above metric can be translated to%
\[
\zeta_{ER}=-\frac{u^{2}}{u^{2}+\mu^{2}}dt\otimes dt+4(u^{2}+\mu^{2})dr\otimes
dr+(u^{2}+\mu^{2})\zeta_{S^{2}}\text{,}%
\]
\emph{holding for all }$u\in\mathbb{R}$\emph{ while }$r\in\lbrack\mu
,\infty\lbrack$. Or, in M. Visser words \cite{Vis} (preserving his notation),

\begin{quotation}
\textquotedblleft This coordinate change discards the region containing the
curvature singularity $r\in\lbrack0,2M)$, and twice covers the asymptotically
flat region, $r\in\lbrack2M,\infty)$. The region near $u=0$ is interpreted as
a \textquotedblleft bridge\textquotedblright\ connecting the asymptotically
flat region near $u=+\infty$ with the asymptotically flat region near
$u=-\infty$\textquotedblright.
\end{quotation}

From a mathematical point of view, this is of course a \textit{non sequitur}
since $r\longmapsto u^{2}=r-\mu$ is not a diffeomorphism $[\mu,\infty
\lbrack\longrightarrow\mathbb{R}$.

We can make the construction of the ER bridge precise as follows.
Let\footnote{$\approx$ means: homeomorphic to.}%
\[
\mathfrak{B=}\text{ }\mathbb{R}\times S^{2}\approx\mathbb{R}\times\{0\}\times
S^{2}\text{,}%
\]%
\[
\mathcal{N}_{1}=\mathbb{R\times}[\mu,\infty\lbrack\times S^{2},\text{
\ \ }\mathcal{N}_{2}=\mathbb{R\times}]-\infty,-\mu]\times S^{2}\text{.}%
\]

Define on these manifolds the pseudo-Riemannian structures $(\mathcal{N}%
_{1},\zeta_{ER}|\mathcal{N}_{1})$, \newline$(\mathcal{N}_{2},\zeta
_{ER}|\mathcal{N}_{2})$ and $(\mathfrak{B},\zeta_{ER}|\mathfrak{B})$. As
$\mathfrak{B}$ can be identified with the boundaries of $\mathcal{N}_{1}$ and
$\mathcal{N}_{2}$, letting $\operatorname{id}_{\mathfrak{B}}%
:\mathfrak{B\longrightarrow B}$ be the identity mapping, we can state our

\begin{definition}
Let $G=(\mathcal{N}_{1},\mathcal{N}_{2},\mathfrak{B\subset}$ $\mathcal{N}%
_{1},\mathfrak{B\subset}$ $\mathcal{N}_{2},\operatorname{id}_{\mathfrak{B}})$
be a pseudo-Riemannian gluing. Then the ER spacetime manifold is the glued
space $Q_{G}$, $\mathfrak{B}$ is called the ER bridge and $\mathcal{N}_{1}$
and $\mathcal{N}_{2}$ the exterior regions.
\end{definition}

In order to understand the bridge geometry, we state our

\begin{proposition}
\label{Bridge}The ER bridge $\mathfrak{B}$ is mapped onto $S^{2}$ by a
homothety with coefficient $\mu$.
\end{proposition}

\begin{proof}
The restriction $\zeta_{ER}|\mathfrak{B}$ is $4\mu^{2}dr\otimes dr+\mu
^{2}\zeta_{S^{2}}$ since $u|\mathfrak{B}=0$. But because $r|\mathfrak{B}=\mu$,
$d(r|\mathfrak{B)}=0$. Thus the former metric becomes $\mu^{2}\zeta_{S^{2}}$.
\end{proof}

\bigskip

It is necessary some care in order to interpret the latter Proposition. As it
was defined above, the bridge $\mathfrak{B}$ has the topology of
$\mathbb{R}\times S^{2}$, where the real line $\mathbb{R}$ represents
physically the time. However, since the metric degenerates on $\mathfrak{B}$
in such a way that the metric component accompanying $dt\otimes dt$ vanishes,
the natural projection $\mathbb{R}\times S^{2}\longrightarrow$ $S^{2}%
,(t,\theta)\mapsto\theta$ becomes a homothety when applied to the metric
tensor of $\mathfrak{B}$, mapping $\mathfrak{B}$ as a pseudo-Riemannian
structure onto the 2-dimensional sphere with radius $\mu$.

In this sense, one might think about the solution mass $\mu$ as being the
radius of the \textquotedblleft throat\textquotedblright\ of the ER bridge.

\bigskip

Now we devote some words to comment on the relation of the Horizon living in
the Kruskal-Szekeres spacetime to the ER Bridge. In his original work, Kruskal
himself (cf. ref. \cite{Kru}) understood the Horizon as a kind of bridge or a
\textquotedblleft wormhole\textquotedblright\ in the sense of Misner and
Wheeler. And indeed, from a mathematical viewpoint, Lemma \ref{Horizon} and
Proposition \ref{Bridge} shows that the Horizon and the ER bridge not only
shares the same topology as they are both mapped onto $S^{2}$ by a homothety,
whose coefficient equals the solution mass.

On the other hand, it is clear that the construction of the Horizon can be
regarded as more \textquotedblleft natural\textquotedblright\ than the ER
bridge, in the sense that while the latter is based on an \textit{ad hoc}
gluing of manifolds, the former is a necessary consequence of the maximal
extension of the Hilbert-Droste manifold.

However, we shall admire the creativity and the originality of Einstein and
Rosen in anticipating some aspects of the KS spacetime three decades before
the publications of Kruskal and Szekeres.

\bigskip

We shall finish this Appendix showing that it is not a good idea to regard the
Horizon or even the bridge $\mathfrak{B}$ as a \textquotedblleft
wormhole\textquotedblright\ since it is impossible to use these objects to
travel from one region to another.

To prove this, it is clearly sufficient to consider only the motion of
lightlike geodesics\footnote{To see why, make a sketch of the
\textquotedblleft local\textquotedblright\ lightcones at some isolated points
of the Kruskal-Szekeres plane and recall that the tangent vector of timelike
geodesics should stays within that lightcones.}. Let $\gamma$ be a null
geodesic into $\mathcal{R}_{II}^{+}$ ending in the Horizon (or the ER bridge).
So by Lemma \ref{KruskalGeodesics}, there exists some $\varepsilon
\in\{-1,+1\}$ and a reparametrization of $\gamma$, say, $\{s\in\mathbb{R}%
:\varepsilon s>\mu\}\overset{\gamma}{\longrightarrow}\mathcal{R}_{II}^{+}$,
such that for some HD coordinates,%
\[
r\circ\gamma(s)=\varepsilon s\text{, \ \ }t\circ\gamma(s)=s+\varepsilon\mu
\log\left\vert \mu-\varepsilon s\right\vert \text{.}%
\]
Since we are interested in in-going geodesics, that is, geodesics which falls
in the Horizon or in the ER bridge, we must choose $\varepsilon=-1$.

Using the diffeormophism $\mathbb{R\times]}\mu,\infty\lbrack\overset{\xi
}{\longrightarrow}\mathcal{R}_{II}^{+}$ whose existence is ensured by Lemma
\ref{KruskalLemma}, we can rewrite the above geodesic parametrization in the
Kruskal-Szekeres coordinates. That is, letting%
\begin{align*}
u\circ\xi(t,r)  &  =\sqrt{\left\vert r-\mu\right\vert }\exp\frac{r+t}{2\mu},\\
v\circ\xi(t,r)  &  =\sqrt{\left\vert r-\mu\right\vert }\exp\frac{r-t}{2\mu},
\end{align*}
where $(u,v)$ are the natural coordinates of $\mathbb{R}^{2}$ restricted to
$\mathcal{Q}_{K}$, we find the parametrization of $\gamma$ as
\[
u\circ\xi(t,r)=1\text{, \ \ }v\circ\xi(t,r)=-(s+\mu)\exp\frac{-s}{\mu}\text{,}%
\]
holding only for $s\in\{s\in\mathbb{R}:-s>\mu\}$. However, because the above
equations are solutions of the geodesic differential equation, the uniqueness
of the ODE theory ensure that there exists only one analytical extension of
these expressions, which is also given\footnote{The reader may verify it
without the ODE theory by using the diffeormophism $\mathbb{R\times]}%
0,\mu\lbrack\longrightarrow\mathcal{R}_{I}^{+}$ corresponding to $\xi$. Recall
the proof of Lemma \ref{KruskalLemma}.} by the latter equations however
holding for all $s\in\{s\in\mathbb{R}:-s>0\}$.

Lastly, because $r$ is defined implicitly in the KS plane by $f(r)=uv$ (cf.
Lemma \ref{KruskalLemma}),%
\[
f\circ r\circ\gamma(s)=u\circ\xi(t,r)v\circ\xi(t,r)=-(s+\mu)\exp\frac{-s}{\mu
}\text{,}%
\]
so that $f\circ r\circ\gamma(0)=-\mu$. Therefore, we conclude that the image
of $\gamma$ always contains points in the black whole, and consequently
$\gamma$ ends in the \textquotedblleft fundamental
singularity\textquotedblright\ at \textquotedblleft$r=0$\textquotedblright%
\ instead of crossing to $\mathcal{R}_{II}^{-}$.

\end{document}